\documentclass[conference,a4paper]{IEEEtran}

\usepackage{amsthm,mathtools,amssymb,mathrsfs}
\usepackage{color}
\usepackage{tikz}
\usepackage[format=hang]{subfig}
\usepgflibrary{shapes}
\usetikzlibrary{petri,arrows,positioning}

\newcommand{\calF}{\mathcal{F}}
\newcommand{\calM}{\mathcal{M}}
\newcommand{\calP}{\mathcal{P}}

\newcommand{\calZ}{\mathcal{Z}}

\newcommand{\A}{\mathcal{A}}
\newcommand{\B}{\mathcal{B}}
\newcommand{\calB}{\mathcal{B}}

\newcommand{\D}{\mathcal{D}}

\newcommand{\calE}{\mathcal{E}}
\newcommand{\F}{\mathcal{F}}

\newcommand{\Z}{\mathcal{Z}}

\newcommand{\calO}{\mathcal{O}}

\newcommand{\V}{\mathcal{V}}

\newcommand{\M}{\mathcal{M}}

\newcommand{\X}{\mathcal{X}}

\newcommand{\Prob}{\mathit{Prob}}

\newcommand{\Nset}{\mathbb{N}}

\newcommand{\Zset}{\mathbb{Z}}
\newcommand{\Qset}{\mathbb{Q}}
\newcommand{\Rset}{\mathbb{R}}
\newcommand{\fpath}{\mathit{FPath}}
\newcommand{\run}{\mathit{Run}}

\newcommand{\len}{\mathit{length}}

\newcommand{\change}{\mathit{change}}

\newcommand{\rt}{\vec{t}_{oc}}
\renewcommand{\vec}[1]{\mathit{\pmb{#1}}}
\newcommand{\pmin}{p_{\mathit{min}}}
\newcommand{\eps}{\varepsilon}
\newcommand{\amax}{a_{\mathit{max}}}

\newcommand{\mindiv}{\mathit{mindiv}}
\newcommand{\Div}{\mathit{Div}}
\newcommand{\qup}{q_{\mathit{up}}}
\newcommand{\qdown}{q_{\mathit{down}}}

\newcommand{\Over}{\mathit{Over}}

\newcommand{\Zall}{\Z_{all}}
\newcommand{\Zminusi}[1]{\Z_{{-}#1}}
\newcommand{\runzc}[2]{\run(#1,#2)}
\newcommand{\bottomfin}{\mathit{botfin}}
\newcommand{\bottominf}{\mathit{botinf}}
\newcommand{\totalrew}[4]{\mathit{tot}_{#1}(#2;#4)}
\newcommand{\cval}{\mathit{cval}}
\newcommand{\state}{\mathit{state}}

\newcommand{\Zminusid}{\{1,\ldots,d-1\}}
\newcommand{\uin}{\vec{u}^{in}}

\newcommand{\msi}[1]{m^{(#1)}_i}

\newcommand{\ms}[1]{m^{(#1)}}

\newcommand{\ps}[1]{p^{(#1)}}
\newcommand{\xs}[1]{x^{(#1)}}

\newcommand{\emax}{e_{\mathit max}}
\newcommand{\gmax}{g_{\mathit max}}
\newcommand{\gmin}{g_{\mathit min}}

\newcommand{\ymin}{y_{\mathit min}}

\newcommand{\tran}[1]{{}\mathchoice%
    {\stackrel{#1}{\rightarrow}}
    {\mathop {\smash\rightarrow}\limits^{\vrule width 0pt height 0pt
                                                depth 4pt\smash{#1}}}
    {\stackrel{#1}{\rightarrow}}
    {\stackrel{#1}{\rightarrow}}
{}}
\newcommand{\ltran}[1]{{}\mathchoice%
    {\stackrel{#1}{\longrightarrow}}
    {\mathop {\smash\longrightarrow}\limits^{\vrule width 0pt height 0pt
                                                depth 4pt\smash{#1}}}
    {\stackrel{#1}{\longrightarrow}}
    {\stackrel{#1}{\longrightarrow}}
{}}
\newcommand{\btran}[1]{{}\mathchoice%
    {\stackrel{#1}{\hookrightarrow}}
    {\mathop {\smash\hookrightarrow}\limits^{\vrule width 0pt height 0pt
                                                depth 4pt\smash{#1}}}
    {\stackrel{#1}{\hookrightarrow}}
    {\stackrel{#1}{\hookrightarrow}}
{}}

\theoremstyle{plain}
\newtheorem{lemma}{Lemma}

\newtheorem{proposition}{Proposition}
\newtheorem{definition}{Definition}

\newtheorem{theorem}{Theorem}

\newtheorem{observation}{Observation}

\newenvironment{reftheorem}[2]{\begin{trivlist}
\item[\hskip \labelsep {\bfseries #1}\hskip \labelsep {\bfseries #2}]\itshape}{\end{trivlist}}

\begin{document}

\title{Zero-Reachability in Probabilistic Multi-Counter Automata}

\author{\IEEEauthorblockN{%
Tom\'{a}\v{s} Br\'{a}zdil$^{1,}$\IEEEauthorrefmark{1},\quad
Stefan Kiefer$^{2,}$\IEEEauthorrefmark{3},\quad
Anton\'{\i}n Ku\v{c}era$^{1,}$\IEEEauthorrefmark{1},\quad
Petr Novotn\'{y}$^{1,}$\IEEEauthorrefmark{1},\quad
Joost-Pieter Katoen\IEEEauthorrefmark{2}}\\
\IEEEauthorblockA{\IEEEauthorrefmark{1}%
Faculty of Informatics, Masaryk University, Brno, Czech Republic\\
\{brazdil,kucera\}@fi.muni.cz, petr.novotny.mail@gmail.com}\smallskip
\IEEEauthorblockA{\IEEEauthorrefmark{2}%
Department of Computer Science, RWTH Aachen University, Germany\\
katoen@cs.rwth-aachen.de}\smallskip
\IEEEauthorblockA{\IEEEauthorrefmark{3}%
Department of Computer Science, University of Oxford, United Kingdom\\
stefan.kiefer@cs.ox.ac.uk}
\IEEEcompsocitemizethanks{\IEEEcompsocthanksitem M.
Shell is with the Georgia Institute of Technology.
\IEEEcompsocthanksitem J. Doe and J. Doe are with An
onymous University.}}

\maketitle

\footnotetext[1]{T.~Br\'{a}zdil, A.~Ku\v{c}era, and P.~Novotn\'{y} are supported
by the Czech Science Foundation, Grant No.~P202/10/1469.}
\footnotetext[2]{S. Kiefer is supported by a Royal Society University Research Fellowship.}

\begin{abstract}
We study the qualitative and quantitative zero-reachability problem
in probabilistic multi-counter systems. We identify the undecidable
variants of the problems, and then we concentrate on the remaining
two cases. In the first case, when we are interested in the
probability of all runs that visit zero in \emph{some} counter,
we show that the qualitative zero-reachability is decidable in time
which is polynomial in the size of a given pMC and doubly exponential
in the number of counters. Further, we show that the probability of all
zero-reaching runs can be effectively approximated up to an arbitrarily
small given error $\varepsilon > 0$ in time which is polynomial
in $\log(\varepsilon)$, exponential in the size of a given pMC, and doubly
exponential in the number of counters. In the second case, we are interested
in the probability of all runs that visit zero in some counter different
from the last counter. Here we show that the qualitative zero-reachability
is decidable and \textsc{SquareRootSum}-hard, and the probability of all
zero-reaching runs can be effectively approximated up to an arbitrarily
small given error $\varepsilon > 0$ (these result
applies to pMC satisfying a suitable technical condition that can
be verified in polynomial time). The proof techniques invented in the
second case allow to construct counterexamples for some classical
results about ergodicity in stochastic Petri nets.

\end{abstract}

\section{Introduction}
\label{sec-intro}

A \emph{probabilistic multi-counter automaton (pMC)} $\A$ of dimension $d \in
\Nset$ is an abstract fully probabilistic computational device
equipped with a finite-state control unit and $d$~unbounded counters
that can store non-negative integers. A \emph{configuration} $p\vec{v}$ of
$\A$ is given by the current control state $p$ and the vector of
current counter values $\vec{v}$. The dynamics of $\A$ is defined
by a finite set of \emph{rules} of the form $(p,\alpha,c,q)$ where $p$ is the
current control state, $q$ is the next control state, $\alpha$ is 
a $d$-dimensional vector of counter changes ranging over $\{-1,0,1\}^d$,
and $c$ is a subset of counters that are tested for zero. Moreover,
each rule is assigned a positive integer \emph{weight}. A rule
$(p,\alpha,c,q)$ is \emph{enabled} in a configuration $p\vec{v}$ if 
the set of all counters with zero value in $\vec{v}$ is precisely~$c$ 
and no component of $\vec{v} + \alpha$ is negative; such an enabled
rule can be \emph{fired} in $p\vec{v}$ and generates a \emph{probabilistic
transition} $p\vec{v} \tran{x} q(\vec{v}{+}\alpha)$ where the probability
$x$ is equal to the weight of the rule divided by the total weight of all
rules enabled in $p\vec{v}$. 
A special subclass of pMC are \emph{probabilistic
vector addition systems with states (pVASS)}, which are equivalent
to (discrete-time) \emph{stochastic Petri nets (SPN)}. Intuitively,
a pVASS is a pMC where no subset of counters is tested for
zero explicitly (see Section~\ref{sec-prelim} for a precise definition).

The decidability and complexity of basic qualitative/quantitative
problems for pMCs has so far been studied mainly in the one-dimensional
case, and there are also some results about unbounded SPN (a more
detailed overview of the existing results is given below).
In this paper, we consider \emph{multi-dimensional} pMC and the associated
\emph{zero-reachability} problem. That is, we are interested in the
probability of all runs initiated in a given $p\vec{v}$
that eventually visit a ``zero configuration''. Since there are 
several counters, the notion of ``zero configuration'' can be formalized 
in various ways (for example, we might want to have zero in some counter, 
in all counters simultaneously, or in a given subset of counters).
Therefore, we consider a general \emph{stopping criterion} $\calZ$
which consists of \emph{minimal} subsets of counters that are required
to be simultaneously zero. For example, if 
$\calZ = \Zall = \{\{1\},\ldots,\{d\}\}$, then a run is stopped when
reaching a configuration with zero in \emph{some} counter; and if 
we put $\calZ = \{\{1,2\}\}$,
then a run is stopped when reaching a configuration with zero in 
counters~$1$ and~$2$ (and possibly also in other counters).
We use $\calP(\run(p\vec{v},\calZ))$ to denote the probability of
all runs initiated in $p\vec{v}$ that reach a configuration satisfying
the stopping criterion~$\calZ$. The main algorithmic problems considered
in this paper are the following:
\begin{itemize}
\item \emph{Qualitative $\calZ$-reachability:} 
   Is $\calP(\run(p\vec{v},\calZ)) = 1$? 
\item \emph{Approximation:} Can 
   $\calP(\run(p\vec{v},\calZ))$ be approximated up 
   to a given absolute/relative 
   error \mbox{$\varepsilon > 0$}?
\end{itemize}
We start by observing that the above problems are not effectively solvable
in general, and we show that there are only two potentially decidable
cases, where $\calZ$ is equal either to $\Zall$ (Case I) or to
$\Zminusi{i} = \Zall \smallsetminus \{\{i\}\}$ (Case II). Recall that if
$\calZ = \Zall$, then a run is stopped when some counter reaches
zero; and if  $\calZ = \Zminusi{i}$, then a run is stopped when a counter
different from~$i$ reaches~zero. Cases~I and~II are analyzed independently
and the following results are achieved:

\textbf{Case I}: We show that the qualitative $\Zall$-reachability problem
is decidable in time polynomial in $|\A|$ and doubly exponential in~$d$.
In particular, this means that the problem is decidable in 
\emph{polynomial time for every fixed $d$}. Then, we show that
$\calP(\run(p\vec{v},\Zall))$ can be effectively approximated up to
a given absolute/relative error \mbox{$\varepsilon > 0$} in time which
is polynomial in $|\varepsilon|$, exponential in $|\A|$, and doubly exponential 
in~$d$ (in the special case when $d = 1$, the problem is known to be 
solvable in time polynomial in $|\A|$ and $|\varepsilon|$, see
\cite{ESY:polynomial-time-termination}).

\textbf{Case II}:  We analyze Case~II only under a
technical assumption that counter~$i$ is not critical; roughly 
speaking, this means that counter~$i$ has either a 
tendency to increase or a tendency to decrease when the other counters are 
positive. The problem whether counter~$i$ is critical or not is 
solvable in time polynomial in $|\A|$, so we can efficiently check 
whether a given pMC can be analyzed by our methods.

Under the mentioned  assumption, we show how to construct
a suitable martingale which captures the behaviour of certain runs
in~$\A$. Thus, we obtain a new and versatile tool for analyzing quantitative
properties of runs in multi-dimensional pMC, which is 
more powerful than the martingale of \cite{BKK:pOC-time-LTL-martingale} 
constructed for one-dimensional pMC. Using this martingale and the
results of \cite{BFLZ:VASSz-model-checking-LMCS}, we show that the qualitative 
\mbox{$\Zminusi{i}$-reachability} problem is decidable. We also
show that the problem is \textsc{Square-Room-Sum}-hard, even
for two-dimensional pMC satisfying the mentioned technical assumption.
Further, we show that $\calP(\run(p\vec{v},\Zminusi{i}))$ can be 
effectively approximated up to a given absolute error 
\mbox{$\varepsilon > 0$}. The main reason why we do not provide any
upper complexity bounds in Case~II is a missing upper bound for coverability
in VAS with one zero test (see \cite{BFLZ:VASSz-model-checking-LMCS}).

\begin{figure}[t]
\begin{center}
\begin{tikzpicture}[x=2cm,y=2cm,>=stealth',
every transition/.style={draw,minimum size=6mm},
every place/.style={draw,minimum size=6mm},
bend angle=45]
\node[place] (1) at (0,0) {};
\node[place] (2) at (2,0) {};
\node at (-0.7,0) {counter 1};
\node at (2.7,0) {counter 2};
\node[transition] (middle) at (1,0) {$100$}
 edge[pre] (1)
 edge[pre] (2);
\node[transition] at (0,0.5) {$1$}
 edge[post] (1);
\node[transition] at (2,0.5) {$1$}
 edge[post] (2);
\node[transition] (1d) at (0,-0.5) {$10$}
 edge[pre] (1);
\draw (1d.west) edge[post,bend left] (1);
\draw (1d.east) edge[post,bend right] (1);
\node[transition] (2d) at (2,-0.5) {$10$}
 edge[pre] (2);
\draw (2d.west) edge[post,bend left] (2);
\draw (2d.east) edge[post,bend right] (2);
\node at (-0.3,-.5) {$t_1$};
\node at (2.3,-.5) {$t_2$};
\end{tikzpicture}
\end{center}
\caption{Firing process may not be ergodic.}
\label{fig-SPN}
\end{figure}
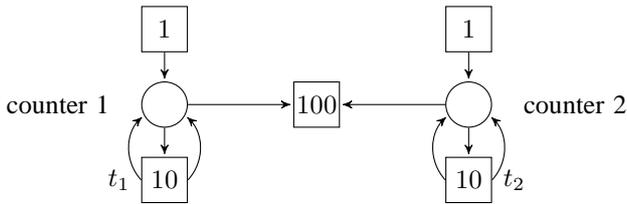

It is worth noting that the techniques developed in Case~II reveal 
the existence of phenomena that should not exist according to the
previous results about ergodicity in SPN. A classical paper in this area
\cite{DBLP:journals/tse/FlorinN89} has been written by Florin \&
Natkin in 80s. In the paper, it is claimed that if the state-space of
a given SPN (with arbitrarily many unbounded places) is strongly
connected, then the firing process is ergodic (see Section IV.B.{}
in \cite{DBLP:journals/tse/FlorinN89}).  In the setting of
discrete-time probabilistic Petri nets, this means that for almost
all runs, the limit frequency of transitions performed along a run is
defined and takes the same value.  However, in Fig.~\ref{fig-SPN} there is an
example of a pVASS (depicted as SPN with weighted transitions) with two  
counters (places) and strongly connected state space 
where the limit frequency of
transitions may take two eligible values (each with probability
$1/2$). Intuitively, if both counters are positive, then both
of them have a tendency to decrease (i.e., the trend of the only BSCC
of $\F_\A$ is negative in both components, see Section~\ref{sec-case1}). 
However, if we reach a configuration
where the first counter is zero and the second counter is sufficiently 
large, then the second counter starts to \emph{increase}, i.e., it
never becomes zero again with some positive probability (cf.{} the
\textit{oc-trend} of the only BSCC~$D$ of $\calB_1$ introduced
in Section~\ref{sec-case2}). The first counter stays zero for most
of the time, because when it becomes positive, it is immediatelly
emptied with a very large probability. This means
that the frequency of firing $t_2$ will be much higher than the frequency
of firing $t_1$. When we reach a configuration where the first counter
is large and the second counter is zero, the situation is symmetric, i.e.,
the frequency of firing $t_1$ becomes much higher than the frequency
of firing $t_2$. Further, almost every run eventually behaves according to
one the two scenarios, and therefore there are two eligible limit 
frequencies of transitions, each of which is taken with probability~$1/2$. 
So, we must unfortunately conclude that the results
of \cite{DBLP:journals/tse/FlorinN89} are not valid for general SPN. 
\smallskip

\noindent
\textbf{Related Work.} 
One-dimensional pMC and their extensions into decision processes
and games were studied in 
\cite{BBEKW:OC-MDP,EWY:one-counter,BKK:pOC-time-LTL-martingale,ESY:polynomial-time-termination,BBEK:OC-games-termination-approx,EWY:one-counter-PE,BBE:OC-games}. In particular, in \cite{ESY:polynomial-time-termination} it was shown that
termination probability (a ``selective'' variant of zero-reachability)
 in one-dimensional pMC can be approximated up to an arbitrarily small 
given error in polynomial time. In \cite{BKK:pOC-time-LTL-martingale},
it was shown how to construct a martingale for a given one-dimensional
pMC which allows to derive tail bounds on termination time (we use
this martingale in Section~\ref{sec-case1}). 

There are also many 
papers about SPN (see, e.g., \cite{DBLP:journals/tc/Molloy82,%
DBLP:journals/tocs/MarsanCB84}), and some of these works also
consider algorithmic aspects of unbounded SPN (see, e.g.,
\cite{AHM:decisive-Markov-chains,DBLP:journals/jss/FlorinN86,%
DBLP:journals/tse/FlorinN89}). 

Considerable amount of papers has been devoted to algorithmic analysis
of so called probabilistic lossy channel systems (PLCS) and their game
extensions
(see~e.g.~\cite{IN:probLCS-TACAS,BE:probLCS-algorithms-ARTS,AHMS:Eager-limit,ACMS:stoch-parity-games-lossy-QEST,ABRS:IC}). PLCS
are a stochastic extension of lossy channel systems, i.e., an
infinite-state model comprising several interconnected queues coupled
with a finite-state control unit. The main ingredient, which makes
results about PLCS incomparable with our results on pMCs, is that 
queues may lose messages with a fixed
loss-rate, which substantially simplifies the associated analysis.

\section{Preliminaries}
\label{sec-prelim}

\noindent
We use $\Zset$, $\Nset$, $\Nset^+$, $\Qset$, and $\Rset$
to denote the set of all integers, non-negative
integers, positive integers, rational numbers, and real numbers, 
respectively. 

Let $\V = (V,L,\tran{})$, where $V$ is a non-empty set of vertices,
$L$ a non-empty set of \emph{labels}, and
${\tran{}} \subseteq V \times L \times V$ a \emph{total} relation
(i.e., for every $v \in V$ there is at least one \emph{outgoing}
transition $(v,\ell,u) \in {\tran{}}$). As usual, we write 
$v \tran{\ell} u$ instead of $(v,\ell,u) \in {\tran{}}$, and
$v \tran{} u$ iff $v \tran{\ell} u$ for some $\ell \in L$.
The reflexive and transitive closure of $\tran{}$ is denoted by $\tran{}^*$. 
A \emph{finite path} in $\V$ of
\emph{length} $k \geq 0$ is a finite sequence
of the form $v_0\ell_0v_1\ell_1\ldots\ell_{k-1} v_k$, where
$v_i \tran{\ell_{i}} v_{i+1}$ for all $0 \leq i <k$. The length of a finite
path $w$ is denoted by $\len(w)$. 
A \emph{run} in $\V$ is an infinite
sequence $w$ of vertices such that every finite prefix of $w$ ending in
a vertex is a finite path in $\V$. The individual vertices of $w$ are 
denoted by $w(0),w(1),\ldots$. The sets of all finite paths and all runs
in $\V$ are denoted by $\fpath_{\V}$ and $\run_{\V}$, respectively.
The sets of all finite paths and all runs in $\V$
that start with a given finite path $w$ are denoted by
$\fpath_{\V}(w)$ and $\run_{\V}(w)$, respectively.
A \emph{strongly connected component (SCC)} of $\V$ is a
maximal subset $C \subseteq V$ such that for all $v,u \in C$
we have that $v \tran{}^* u$. A SCC~$C$ of $\V$ is 
a \emph{bottom SCC (BSCC)} of $\V$ if for all 
$v \in C$ and $u \in V$ such that $v \tran{} u$ we have that
$u \in C$.

We assume familiarity with basic notions of probability theory, e.g.,
\emph{probability space}, \emph{random variable}, or the \emph{expected 
value}. As usual, a \emph{probability distribution} over a finite or
countably infinite set $A$ is a function
$f : A \rightarrow [0,1]$ such that \mbox{$\sum_{a \in A} f(a) = 1$}.
We call $f$ \emph{positive} if
$f(a) > 0$ for every $a \in A$, and \emph{rational} if $f(a) \in
\Qset$ for every $a \in A$.

\begin{definition}
\label{def-Markov-chain}
  A \emph{labeled Markov chain} is a tuple \mbox{$\M = (S,L,\tran{},\Prob)$}
  where $S \neq \emptyset$ is a finite or countably infinite
  set of \emph{states}, $L \neq \emptyset$ is a finite or countably 
  infinite set of \emph{labels},
  \mbox{${\tran{}} \subseteq S \times L \times S$} is a total 
  \emph{transition relation}, and $\Prob$ is a function that assigns 
  to each state $s \in S$
  a positive probability distribution over the outgoing transitions
  of~$s$. We write $s \ltran{\ell,x} t$ when $s \tran{\ell} t$
  and $x$ is the probability of $(s,\ell,t)$.
\end{definition}

\noindent
If $L = \{\ell\}$ is a singleton, we say that $\M$ is 
\emph{non-labeled}, and we omit both $L$ and $\ell$ when specifying $\M$
(in particular, we write $s \tran{x} t$ instead of $s \ltran{\ell,x} t$).
To every $s \in S$ we associate the standard probability space
$(\run_{\M}(s),\calF,\calP)$ of runs starting at $s$, where $\calF$ is
the \mbox{$\sigma$-field} generated by all \emph{basic cylinders}
$\run_{\M}(w)$, where $w$ is a finite path starting at~$s$, and
$\calP: \calF \rightarrow [0,1]$ is the unique probability measure
such that $\calP(\run_{\M}(w)) = \prod_{i{=}1}^{\len(w)} x_i$ where
$x_i$ is the probability of $w(i{-}1) \ltran{\ell_{i-1}} w(i)$ for
every $1 \leq i \leq \len(w)$.  If $\len(w) = 0$, we put
$\calP(\run_{\M}(w)) = 1$.

Now we introduce probabilistic multi-counter automata (pMC). 
For technical convenience, we consider \emph{labeled} rules, 
where the associated finite 
set of labels always contains a distinguished element~$\tau$.
The role of the labels becomes clear in Section~\ref{sec-case2} where
we abstract a (labeled) one-dimensional pMC from a given 
multi-dimensional one.

\begin{definition}
\label{def-pVASS}
Let $L$ be a finite set of labels such that $\tau \in L$, and let 
$d \in \Nset^+$. An $L$-labeled $d$-dimensional 
\emph{probabilistic multi-counter automaton (pMC)} is a 
triple  $\A = (Q,\gamma,W)$, where 
\begin{itemize}
\item $Q$ is a finite set of \emph{states}, 
\item \mbox{$\gamma \subseteq Q \times \{-1,0,1\}^d \times 2^{\{1,\ldots,d\}} 
   \times L \times Q$} is a set of \emph{rules} such that for all
   $p \in Q$ and $c \subseteq \{1,\ldots,d\}$ there is at least one 
   outgoing rule of the form $(p,\vec{\alpha},c,\ell,q)$,
\item $W : \gamma \rightarrow \Nset^+$ is a \emph{weight assignment}.
\end{itemize}
\end{definition}

\noindent
The encoding size of $\A$ is denoted by $|\A|$,
where the weights used in $W$ and the counter indexes used in 
$\gamma$ are encoded in binary. 

A \emph{configuration} of $\A$ is an element of $Q \times \Nset^d$,
written as $p\vec{v}$. We use 
$Z(p\vec{v}) = \{ i \mid 1\leq i \leq d, \vec{v}[i] = 0\}$ 
to denote the set of all counters that are zero in $p\vec{v}$. 
A rule $(p,\vec{\alpha},c,\ell,q) \in \gamma$ is \emph{enabled} in a 
configuration $p\vec{v}$ if $Z(p\vec{v}) = c$ and 
for all $1 \leq i \leq d$ where $\vec{\alpha}[i] = -1$ we have that 
$\vec{v}[i] > 0$.

The semantics of a $\A$ is given by the associated $L$-labeled 
Markov chain $\M_\A$ whose states are the configurations 
of~$\A$, and the outgoing transitions of a configuration 
$p\vec{v}$ are determined as follows:
\begin{itemize}
\item If no rule of $\gamma$ is enabled in $p\vec{v}$, then
  $p\vec{v} \ltran{\tau,1} p\vec{v}$ is the only outgoing transition 
  of $p\vec{v}$;
\item otherwise, for every rule $(p,\vec{\alpha},c,\ell,q) \in \gamma$ enabled in
  $p\vec{v}$ there is a transition  
  $p\vec{v} \ltran{x,\ell} q\vec{u}$ such that 
  $\vec{u} = \vec{v}+\vec{\alpha}$ and 
  $x = W((p,\vec{\alpha},c,\ell,q))/T$, where $T$ is the total weight of all 
  rules enabled in $p\vec{v}$.
\end{itemize}

\noindent
When $L = \{\tau\}$,
we say that $\A$ is \emph{non-labeled}, and both $L$ and $\tau$ 
are omitted when specifying $\A$. We say that $\A$ is a \emph{probabilistic
vector addition system with states (pVASS)} if no subset of counters 
is tested for zero, i.e., for every $(p,\vec{\alpha},\ell,q) \in 
Q \times \{-1,0,1\}^d \times L \times Q$ we have that $\gamma$ contains
either all rules of the form $(p,\vec{\alpha},c,\ell,q)$ (for all 
$c \subseteq \{1,\ldots,d\}$) with the same weight, or no such rule.
For every configuration $p\vec{v}$, %
we use  $\state(p\vec{v})$ and $\cval(p\vec{v})$ to denote the 
control state $p$ and the vector of counter values $\vec{v}$, 
respectively. We also use $\cval_i(p\vec{v})$ to denote~$\vec{v}[i]$.

\medskip

\noindent
\textbf{Qualitative zero-reachability.} 
A \emph{stopping criterion} is a non-empty set 
\mbox{$\calZ \subseteq 2^{\{1,\ldots,d\}}$} of pairwise incomparable
non-empty subsets of counters. 
For every configuration $p\vec{v}$, let $\run(p\vec{v},\calZ)$ be the set 
of all $w \in \run(p\vec{v})$ such that there exist $k \in \Nset$ and
$\varrho \in \calZ$ satisfying $\varrho \subseteq Z(w(i))$. Intuitively, 
$\calZ$ specifies the minimal subsets of counters that must be 
\emph{simultaneously} zero to stop a run.
The \emph{qualitative \mbox{$\calZ$-reachability} problem} is formulated 
as follows: 
\smallskip

\noindent
\textbf{Instance:} A \mbox{$d$-dimensional} pMC $\A$ and a control state 
$p$ of~$\A$.
\textbf{Question:} Do we have $\calP(\run(p\vec{1},\calZ)) = 1$ ? 
\smallskip

\noindent
Here $\vec{1}=(1,\ldots,1)$ is a $d$-dimensional vector of $1$'s.
We also use $\run(p\vec{v},\neg\calZ)$ to denote 
$\run(p\vec{v}) \smallsetminus \run(p\vec{v},\calZ)$, and
we say that $w \in \fpath(p\vec{v})$ is \emph{$\calZ$-safe}
if for all $w(i)$ where $0 \leq i < \len(w)$ and all 
$\varrho \in \calZ$ we have that $\varrho \not\subseteq Z(w(i))$.

\section{The Results}
\label{sec-results}

We start by observing that the qualitative zero-reachability 
problem is undecidable in general, and we identify potentially
decidable subcases. 

\begin{observation}
\label{thm-undecidable}
  Let \mbox{$\calZ \subseteq 2^{\{1,\ldots,d\}}$} be a stopping criterion
  satisfying one of the following conditions:
  \begin{itemize}
  \item[(a)] there is $\varrho \in \calZ$ with more than one element;
  \item[(b)] there are $i,j \in \{1,\ldots,d\}$ such that $i \neq j$ and
    for every $\varrho \in \calZ$ we have that $\{i,j\} \cap \varrho 
    = \emptyset$.
  \end{itemize}
  Then, the qualitative $\calZ$-reachability problem is \emph{undecidable},
  even if the set of instances is restricted to pairs $(\A,p)$
  such that $\calP(\run(p\vec{1},\calZ))$ is either~$0$ or~$1$ (hence, 
  $\calP(\run(p\vec{1},\calZ))$ cannot be effectively approximated 
  up to an absolute error smaller than~$0.5$).
\end{observation}

\noindent
A proof of Observation~\ref{thm-undecidable} is immediate. 
For a given Minsky machine $M$ (see \cite{Minsky:book}) with two counters 
initialized to one,
we construct pMCs $\A_a$ and $\A_b$ of dimension $2$ and $3$, respectively,
and a control state $p$ such that
\begin{itemize}
\item if $M$ halts, then $\calP(\run_{\M_{\A_a}}(p\vec{1},\{\{1,2\}\})) = 1$
   and $\calP(\run_{\M_{\A_b}}(p\vec{1},\{\{3\}\})) = 1$;
\item if $M$ does not halt, then 
   $\calP(\run_{\M_{\A_a}}(p\vec{1},\{\{1,2\}\})) = 0$
   and $\calP(\run_{\M_{\A_b}}(p\vec{1},\{\{3\}\})) = 0$.
\end{itemize}
The construction of $\A_a$ and $\A_b$ is trivial (and hence omitted).
Note that $\A_b$ can faithfully simulate the instructions of $M$ using
the counters $1$ and $2$. The third counter
is decreased to zero only when a control state corresponding to the halting 
instruction of $M$ is reached. Similarly, $\A_a$ simulates the instructions of 
$M$ using its two counters, but here we need to ensure that a configuration
where \emph{both} counters are simultaneously zero is entered iff 
a control state corresponding to the halting instruction of $M$ is reached. 
This is achieved by increasing both counters by $1$ initially, and then
decreasing/increasing counter~$i$ before/after simulating a given instruction
of $M$ operating on counter~$i$.

Note that the construction of $\A_a$ and $\A_b$ can trivially be adapted
to pMCs of higher dimensions satisfying the conditions~(a) and~(b)
of Observation~\ref{thm-undecidable}, respectively. However,
there are two cases not covered by Observation~\ref{thm-undecidable}: 
\begin{itemize}
\item[I.] $\Zall = \{\{1\},\ldots,\{d\}\}$, i.e., a run is stopped 
   when \emph{some} counter reaches zero.
\item[II.] $\Zminusi{i} = \{\{1\},\ldots,\{d\}\} \smallsetminus \{\{i\}\}$ where
   $i \in \{1,\ldots,d\}$, i.e., a run is stopped when a counter different 
   from~$i$ reaches zero. The counters different from~$i$ are called 
   \emph{stopping counters}.
\end{itemize}
These cases are analyzed in the following subsections.

\subsection{Zero-Reachability, Case~I}
\label{sec-case1}

For the rest of this section, let us fix a (non-labeled) 
pMC $\A = (Q,\gamma,W)$ of dimension $d \in \Nset^+$ and 
a configuration $p\vec{v}$. 

Our aim is to identify the conditions under which
$\calP(\run(p\vec{v},\neg\Zall)) > 0$. To achieve that, we 
first consider a (non-labeled) finite-state Markov chain
$\F_{\A} = (Q,\btran{},\Prob)$ where $q \btran{x} r$ iff 
\[
  x \quad = \quad 
  \sum_{(q,\vec{\alpha},\emptyset,r) \in \gamma} P_\emptyset(q,\vec{\alpha},\emptyset,r) 
  \quad > \quad 0.
\]
Here $P_\emptyset : \gamma \rightarrow [0,1]$ is the probability assignment
for the rules defined as follows (we write $P_\emptyset(q,\vec{\alpha},\emptyset,r)$
instead of $P_\emptyset((q,\vec{\alpha},\emptyset,r))$):
\begin{itemize}
\item For every rule $(p,\vec{\alpha},c,q)$ where $c \neq \emptyset$
  we put $P_\emptyset(p,\vec{\alpha},c,q) = 0$.
\item $P_\emptyset(p,\vec{\alpha},\emptyset,q) = W((p,\alpha,\emptyset,q))/T$,
  where $T$ is the total weight of all rules of the form
  $(p,\vec{\alpha}',\emptyset,q')$.
\end{itemize}
Intuitively, a state $q$ of $\F_{\A}$ captures the behavior of 
configurations $q \vec{u}$ where all components of $\vec{u}$ are
positive. 

Further, we partition the states of 
$Q$ into SCCs $C_1,\ldots,C_m$ according to~$\hookrightarrow$. Note that every 
run $w \in \run(p\vec{v})$ eventually \emph{stays} in precisely
one $C_j$, i.e., there is precisely one $1\leq j \leq m$ 
such that for some $k \in \Nset$, the control state of every $w(k')$,
where $k' \geq k$, belongs to~$C_i$. 
We use $\runzc{p\vec{v}}{C_j}$ to denote 
the set of all $w \in \run(p\vec{v},\neg\Zall)$ that stay in~$C_j$. 
Obviously,
\[
  \run(p\vec{v},\neg\Zall) = 
  \run(p\vec{v},C_1) \uplus \cdots \uplus \runzc{p\vec{v}}{C_m}.
\]

  For any $n \in \Nset$ denote by $P_n$ the probability that a run $w$ initiated in $p\vec{v}$ satisfies the following for every $0\leq i \leq n$: $\state(w(i))$ does not belong to any BSCC of $\F_\A$ and $Z(w(i))=\emptyset$. The following lemma shows that $P_n$ decays exponentially fast.

\begin{lemma}\label{lem:F_A-BSCC}
  For any $n \in \Nset$ we have $$P_n \leq (1 - \pmin^{|Q|})^{\lfloor\frac{n}{|Q|}\rfloor},$$ where $\pmin$ is the minimal positive transition probability in $\M_\A$.
  In particular, for any non-bottom SCC $C$ of $\F_\A$ we have 
  $\calP(\run(p\vec{v},C)) = 0$.
\end{lemma}
\begin{proof}
 The lemma immediately follows from the fact that for every configuration $p\vec{v}$ there is a path (in $\A$) of length at most $|Q|$ to a configuration $q\vec{u}$ satisfying either $Z(q\vec{u})\neq \emptyset$ or $q\in D$ for some BSCC $D$ of $\F_\A$.
\end{proof}

Now, let $C$ be a BSCC of $\F_{\A}$. For every $q \in C$, let 
$\change^q$ be a $d$-dimensional vector of \emph{expected counter changes}
given by 
\[
   \change^q_i = 
   \sum_{(q,\vec{\alpha},\emptyset,r) \in \gamma} 
     P_{\emptyset}(q,\vec{\alpha},\emptyset,r) \cdot \vec{\alpha}[i] \,.
\]
Note that $C$ can be seen
as a finite-state irreducible Markov chain, and hence there exists the
unique \emph{invariant distribution} $\mu$ on the states of~$C$ 
(see, e.g.,~\cite{KS:book}) satisfying
\[
\mu(q)\quad =\quad \sum_{r\btran{x} q} \mu(r)\cdot x \,.
\]
The \emph{trend} of $C$ is a $d$-dimensional vector $\vec{t}$  defined by
\[
  \vec{t}[i] \quad = \quad \sum_{q \in C} \mu(q) \cdot \change^q_i \,.
\]
Further, for every $i \in \{1,\ldots,d\}$ and every $q\in C$, we denote
by $\bottomfin_i(q)$ the \emph{least} $j \in \Nset$ such that 
for every configuration $q \vec{u}$ where $\vec{u}[i] = j$,
there is \emph{no} $w \in \fpath_{\M_\A}(q \vec{u})$ where
counter~$i$ is zero in the last configuration of $w$ and all counters
stay positive in every $w(k)$, where $0 \leq k < \len(w)$.
If there is no such $j$, we put $\bottomfin_i(q) = \infty$.
It is easy to show that if $\bottomfin_i(q) < \infty$, then $\bottomfin_i(q)
\leq |C|$; and if $\bottomfin_i(q) = \infty$, then $\bottomfin_i(r)
= \infty$ for all $r \in C$. Moreover, if $\bottomfin_i(q) < \infty$, then there is a $\Zminusi{i}$-safe finite path of length at most $|C|-1$ from $q\vec{u}$ to a configuration with $i$-th counter equal to 0, where $\vec{u}[i]=\bottomfin_i(q)-1$ and $\vec{u}[\ell]=|C|$ for $\ell\neq i$. In particular, the number $\bottomfin_i(q)$ is computable in time polynomial in $|C|$.

We say that counter~$i$ is
\emph{decreasing} in $C$ if $\bottomfin_i(q) = \infty$ for some (and
hence all) $q \in C$.

\begin{definition}
  Let $C$ be a BSCC of $\F_{\A}$ with trend $\vec{t}$, and let 
  $i \in \{1,\ldots,d\}$. We say that counter~$i$ is \emph{diverging}
  in $C$ if either $\vec{t}[i] > 0$, or $\vec{t}[i] = 0$ and the counter~$i$ 
  is not decreasing in~$C$.
\end{definition}
Intuitively, our aim is to prove that $\calP(\run(p\vec{v},C))>0$ iff
all counters are diverging in $C$ and $p\vec{v}$ can reach a configuration 
$q\vec{u}$ (via a $\Zall$-safe finite path) where all components of 
$\vec{u}$ are ``sufficiently large''.
To analyze the individual counters, for every $i \in \{1,\ldots,d\}$
we introduce a~(labeled) \emph{one-dimensional} pMC
which faithfully simulates the behavior of counter~$i$
and ``updates'' the other counters just symbolically in the labels.

\begin{definition} 
  Let $L = \{-1,0,1\}^{d-1}$, and let $\B_i = (Q,\hat{\gamma},\hat{W})$
  be an $L$-labeled pMC of dimension one such that 
  \begin{itemize}\itemsep1ex
  \item $(q,j,\emptyset,\vec{\beta},r) \in \hat{\gamma}$ \ iff \ 
        $(q,\langle \vec{\beta},j\rangle_i,\emptyset,r) \in \gamma$;
  \item $(q,j,\{1\},\vec{\beta},r) \in \hat{\gamma}$ \ iff \ 
        $(q,\langle \vec{\beta},j\rangle_i,\{i\},r) \in \gamma$;
  \item $\hat{W}(q,j,\emptyset,\vec{\beta},r) = 
         W(q,\langle \vec{\beta},j \rangle_i,\emptyset,r)$.
  \item $\hat{W}(q,j,\{1\},\vec{\beta},r) = 
         W(q,\langle \vec{\beta},j \rangle_i,\{i\},r)$.
  \end{itemize}
  Here, 
  $\langle (j_1,\ldots,j_{d-1}),j\rangle_i = 
   (j_1,\ldots,j_{i-1},j,j_i,\ldots,j_{d-1})$.
\end{definition}

\noindent
Observe that the symbolic updates of the counters different from~$i$ 
``performed'' in the labels of $\B_i$ mimic the real updates performed
by $\A$ in configurations where all of these counters are
positive. 

Given a run $w\equiv p_0(v_0) \, \vec{\alpha}_0 \, 
p_1(v_1) \, \vec{\alpha}_1 \, p_2(v_2) \, \vec{\alpha}_2 \, \ldots$ in $\run_{\M_{\B}}(p_0(v_0))$ and $k\in \Nset$, we denote by $\totalrew{}{w}{0}{k}$ the vector $\sum_{n=0}^{k-1} \vec{\alpha}_n$, and given $j\in \{1,\ldots,d\}\smallsetminus \{i\}$, we denote by $\totalrew{j}{w}{0}{k}$ the number $\sum_{n=0}^{k-1} \vec{\alpha}_n[j]$ (i.e., the $j$-th component of $\sum_{n=0}^{k-1} \vec{\alpha}_n$).

Let $\Upsilon_i$ be a function which for a given run 
$w \equiv p_0\vec{v}_0\,p_1\vec{v}_1\, p_2\vec{v}_2\ldots$ of
$\run_{\M_\A}(p\vec{v},\neg\Zminusi{i})$ returns a run
$\Upsilon_i(w) \equiv p_0(\vec{v}_0[i]) \, \vec{\alpha}_0 \, 
p_1(\vec{v}_1[i]) \, \vec{\alpha}_1 \, p_2(\vec{v}_2[i]) \, \vec{\alpha}_2 \, \ldots$
of $\run_{\M_{\B_i}}(p(\vec{v}[i]))$ where the label $\vec{\alpha}_j$ 
corresponds to the update in the abstracted counters performed in
the transition $p_j\vec{v}_j \tran{} p_{j+1}\vec{v}_{j+1}$,
i.e., $\vec{v}_{j+1} - \vec{v}_j = 
\langle \vec{\alpha}_j,\vec{v}_{j+1}[i] - \vec{v}_j[i] \rangle_i$. 
The next lemma is immediate.
\begin{lemma}\label{prop:one-counter-runs}
For all $w \in \run_{\M_\A}(p\vec{v},\neg\Zminusi{i})$ and 
$k \in \Nset$ we have that
\begin{itemize}
\item $\state(w(k))=\state(\Upsilon_i(w)(k))$,
\item $\cval(w(k)) = \langle \totalrew{}{\Upsilon_i(w)}{0}{k}, 
  \cval_1(\Upsilon_i(w)(k))\rangle_i$.
\end{itemize}
Further, for every measurable set 
$R\subseteq \run_{\M_\A}(p\vec{v},\neg\Zminusi{i})$ we have that $\Upsilon_i(R)$ is measurable and
\begin{equation}
\calP(R) \ = \
  \calP(\Upsilon_i(R)) 
\label{eq-project}
\end{equation}%
\end{lemma}

\noindent
Now we examine the runs of $\run(p\vec{v},C)$ where $C$ is a BSCC of
$\F_\A$ such that some counter is not diverging in~$C$. A proof of
the next lemma can be found in Appendix~\ref{app-sec1}.

\begin{lemma}
\label{lem:not-diverging}
  Let $C$ be a BSCC of $\F_{\A}$.
  If some counter is not diverging in $C$, then $\calP(\run(p\vec{v},C)) = 0$.
\end{lemma}

It remains to consider the case when $C$ is a BSCC of $\F_{\A}$
where all counters are diverging. Here we use the results of
\cite{BKK:pOC-time-LTL-martingale} which allow to derive
a bound on divergence probability in one-dimensional pMC.
These results are based on designing and analyzing a suitable
martingale for one-dimensional pMC. 

\begin{lemma}
\label{lem-divergence}
  Let $\B$ be a $1$-dimensional pMC, let $C$ be a BSCC of $\F_\B$ 
  such that the trend $t$ of the only counter in $C$ is positive and let $\delta=2|C|/x_{\min}^{|C|}$ where $x_{\min}$ is the smallest non-zero transition probability in $\M_\B$.
  Then for all $q \in C$ and $k > 2\delta/t$ we have that  
  $\calP(q(k),\neg\Z) \geq 1-\left(a^k/(1+a)\right)$, where $\Z = \{1\}$ and $a=\exp\left(-t^2\, /\, 8(\delta+t+1)^2\right)$.
\end{lemma}
\begin{proof}
Denote by $[q(k){\downarrow},\ell]$ the probability that a run initiated in $q(k)$ visits a configuration with zero counter value for the first time in exactly $\ell$ steps. By Proposition~7 of \cite{BKK:pOC-time-LTL-martingale-arxiv} we obtain for all $\ell\geq h = 2\delta/t$~\footnote{The precise bound on $h$ is given in Proposition~7~\cite{BKK:pOC-time-LTL-martingale-arxiv}.},
\[
[q(k){\downarrow},\ell]\quad \leq\quad a^\ell
\]
where $a=\exp\left(-t^2\, /\, 8(\delta+t+1)^2\right)$ for $\delta\leq 2|C|/x_{\min}^{|C|}$~\footnote{The bound on $\delta$ is given in Proposition 6~\cite{BKK:pOC-time-LTL-martingale-arxiv}.}.

Thus
\[
\calP(q(k),\neg\Z)\geq 1-\sum_{\ell=k}^{\infty} [q(k){\downarrow},\ell]=1-\frac{a^k}{1+a}
\]
\end{proof}

\begin{definition}
  Let $C$ be a BSCC of $\F_{\A}$ where all counters are diverging,
  and let $q \in C$. We say that a configuration $q\vec{u}$
  is \emph{above} a given $n \in \Nset$ if $\vec{u}[i] \geq n$ for every
  $i$ %
  such that $\vec{t}[i] > 0$, and 
  $\vec{u}[i] \geq \bottomfin_i(q)$ for every
  $i$ %
  such that $\vec{t}[i] = 0$. 
\end{definition}

\begin{lemma}
\label{lem-diverging}
  Let $C$ be a BSCC of $\F_{\A}$ where all counters are diverging.
  Then $\calP(\run(p\vec{v},C)) > 0$ iff there is a $\Zall$-safe finite 
  path of the form $p\vec{v} \tran{}^* q\vec{u} \tran{}^* q\vec{z}$
  where $q \in C$, $q\vec{u}$ is above $1$, $\vec{z} - \vec{u} \geq \vec{0}$,
  and $(\vec{z} - \vec{u})[i] > 0$ for every $i$ such that $\vec{t}[i] > 0$. 

\end{lemma}
\begin{proof}
  We start with ``$\Rightarrow$''. 
  Let $\vec{t}$ be the trend of $C$. We show that for almost
  all $w \in \run(p\vec{v},C)$ and all $i \in \{1,\ldots,d\}$, one of
  the following conditions holds:
  \begin{enumerate}
  \item[(A)] $\vec{t}[i]>0$ and $\liminf_{k\rightarrow \infty} \cval_i(w(k))=\infty$,
  \item[(B)] $\vec{t}[i]=0$ and $\cval_i(w(k))\geq \bottomfin_i(\state(w(k)))$ 
    for all $k$'s large enough.
  \end{enumerate}
  First, recall that $C$ is also a BSCC of $\F_{\B_i}$, and realize that
  the trend of the (only) counter in the BSCC $C$ of $\F_{\B_i}$ is~$\vec{t}[i]$.

  Concerning~(A), it follows, e.g., from the results of
  \cite{BKK:pOC-time-LTL-martingale}, that almost all runs
  $w' \in \run_{\M_{\B_i}}(p(1))$ that stay
  in $C$ and do not visit a configuration with zero counter 
  satisfy $\liminf_{k\rightarrow \infty} \cval_1(w'(k))=\infty$.
  In particular, this means that almost all 
  $w' \in \Upsilon_i(\run(p\vec{v},C))$
  satisfy this property.
  Hence, by Lemma~\ref{prop:one-counter-runs}, for almost all 
  $w \in \run(p\vec{v},C)$ we have that
  \mbox{$\liminf_{k\rightarrow \infty} \cval_i(w(k))=\infty$}.

  Concerning (B), note that almost all runs $w \in \run(p\vec{v},C)$
  satisfying $\cval_i(w'(k)) < \bottomfin_i(\state(w(k)))$
  for infinitely many $k$'s eventually visit zero in some counter 
  (there is a path of length at most $|C|$ from each such
  $w(k)$ to a configuration with zero in counter $i$, or in one of the
  other counters).  

  The above claim immediately implies that for every $k \in \Nset$, 
  almost every run of $\run(p\vec{v},C)$ visits a configuration 
  $q\vec{u}$ above~$k$. Hence, there must be a $\Zall$-safe 
  path of the form $p\vec{v} \tran{}^* q\vec{u} \tran{}^* q\vec{z}$
  with the required properties.

  ``$\Leftarrow$'': If there is a $\Zall$-safe 
  path of the form $p\vec{v} \tran{}^* q\vec{u} \tran{}^* q\vec{z}$
  where $q \in C$, $q\vec{u}$ is above $1$, $\vec{z} - \vec{u} \geq \vec{0}$,
  and $(\vec{z} - \vec{u})[i] > 0$ for every $i$ such that $\vec{t}[i] > 0$, then
  $p\vec{v}$ can a reach a configuration $q\vec{y}$ above~$k$ for an 
  arbitrarily large $k \in \Nset$ via a $\Zall$-safe path.

  By Lemma~\ref{lem-divergence}, there exists
  $k \in \Nset$ such that for every $i \in \{1,\ldots,d\}$ where 
  $\vec{t}[i] > 0$ and every $n \geq k$,
  the probability of all $w \in \run_{\M_{\B_i}}(q(n))$ that 
  visit a configuration with zero counter is strictly smaller than 
  \mbox{$1/d$}. Let $q\vec{y}$ be a configuration above~$k$ reachable
  from $p\vec{v}$ via a $\Zall$-safe path (the existence of such 
  a $q\vec{y}$ follows from the existence of
  $p\vec{v} \tran{}^* q\vec{u} \tran{}^* q\vec{z}$). It suffices
  to show that $\calP(\run(q\vec{y},\Zall)) < 1$. For every 
  $i \in  \{1,\ldots,d\}$ where $\vec{t}[i] > 0$,
  let $R_i$ be the set of all $w \in \run(q\vec{y},\Zall)$ such that
  $\cval_i(w(k)) = 0$ for some $k \in \Nset$ and
  all counters stay positive in all $w(k')$ where $k' < k$. 
  Clearly, $\run(q\vec{y},\Zall) = \bigcup_{i} R_i$, and thus we obtain
  \begin{equation*}
    \calP(\run(q\vec{y},\Zall)) \leq \sum_{i} \calP(R_i) =
    \sum_{i} \calP(\Upsilon_i(R_i)) <  d \cdot \frac{1}{d} = 1
  \end{equation*}
\end{proof}

The following lemma shows that it is possible to decide, whether for a given $n\in\Nset$ a configuration above $n$ can be reached via a $\Zall$-safe path. Its proof uses the results of~\cite{BG:VASS-coverability} on the coverability problem in (non-stochastic) VASS.

\begin{lemma}
\label{lem:cover-short-path}
 Let $C$ be a BSCC of $\F_{\A}$ where all counters are diverging and let $q\in C$. There is a $\Zall$-safe finite path of the form $p\vec{v}\tran{}^* q\vec{u}$ with $q\vec{u}$ is above some $n\in \Nset$ iff there is a $\Zall$-safe finite path of length at most $(|Q|+|\gamma|)\cdot(3+n)^{(3d)!+1}$ of the form $p\vec{v}\tran{}^* q\vec{u'}$ with $q\vec{u'}$ is above $n$. %
 Moreover, the existence of such a path can be decided in time $(|\A|\cdot n)^{c'\cdot 2^{d\log(d)}}$ where $c'$ is a fixed constant independent of $d$ and $\A$.
\end{lemma}
\begin{proof}
  We employ a decision procedure of \cite{BG:VASS-coverability} for
  VASS coverability. Since we need to reach $q\vec{u'}$ above $n$ via
  a $\Zall$-safe finite path, we transform $\A$ into a
  (non-probabilistic) VASS $\A'$ whose control states and rules are
  determined as follows: for every rule $(p,\vec{\alpha},\emptyset,q)$
  of $\A$, we add to $\A'$ the control states $p,q$ together with two
  auxiliary fresh control states $q',q''$, and we also add the rules
  $(p,\vec{-1},q')$, $(q',\vec{1},q'')$,
  $(q'',\vec{\alpha},q)$. Hence, $\A'$ behaves like $\A$, but when
  some counter becomes zero, then $\A'$ is stuck (i.e., no transition
  is enabled except for the self-loop). Now it is easy to check that
  $p\vec{v}$ can reach a configuration $q\vec{u}$ above $n$ via a
  $\Zall$-safe finite path in $\A$ iff $p\vec{v}$ can reach a
  configuration $q\vec{u}$ above $n$ via \emph{some} finite path in
  $\A'$, which is exactly the coverability problem for VASS. 
  Theorem~1 in~\cite{BG:VASS-coverability} shows that such a
  configuration can be reached iff there is configuration $q\vec{u'}$
  above $n$ reachable via some finite path of length at most $m =
  (|Q|+|\gamma|)\cdot(3+n)^{(3d)!+1}$. (The term $(|Q|+|\gamma|)$
  represents the number of control states of $\A'$.) This path
  induces, in a natural way, a $\Zall$-safe path from $p\vec{v}$ to
  $q\vec{u'}$ in $\A$ of length at most $m/2$. Moreover, Theorem~2
  in~\cite{BG:VASS-coverability} shows that the existence of such a
  path in $\A'$ can be decided in time
  $(|Q|+|\gamma|)\cdot(3+n)^{2^{\mathcal{O}(d\log(d))}}$, which proves
  the lemma. %
\end{proof}

\begin{theorem}
\label{thm:qual-all-algorithm}
  The qualitative $\Zall$-reachability problem for \mbox{$d$-dimensional}
  pMC is decidable in time $|\A|^{\kappa \cdot 2^{d\log(d)}}$, where $\kappa$ is 
  a fixed constant independent of $d$ and $\A$. 
\end{theorem}
\begin{proof}
  Note that the Markov chain $\F_\A$ is computable in time polynomial
  in $|\A|$ and $d$, and we can efficiently identify all diverging
  BSCCs of $\F_\A$. For each diverging BSCC $C$, we need to check the
  condition of Lemma~\ref{lem-diverging}. By applying Lemma~2.3.{}
  of \cite{RY:VASS-JCSS}, we obtain that if there exist \emph{some} $q\vec{u}$ 
  above~$1$ and a $\Zall$-safe finite path of the form 
  $q\vec{u} \tran{}^* q\vec{z}$ such that $\vec{z} - \vec{u} \geq \vec{0}$
  and $(\vec{z} - \vec{u})[i] > 0$ for every $i$ where $\vec{t}[i] > 0$,
  then such a path exists for \emph{every} $q\vec{u}$ above
  $|\A|^{c\cdot d}$ and its length is bounded by $|\A|^{c\cdot d}$. Here
  $c$ is a fixed constant independent of $|\A|$ and~$d$ (let us
  note that Lemma~2.3.{} of \cite{RY:VASS-JCSS} is formulated for 
  vector addition
  systems without states and a non-strict increase in every counter, 
  but the corresponding result for VASS is easy
  to derive; see also Lemma~15 in \cite{BJK:VASS-games-arxiv}). 
  Hence, the existence
  of such a path for a given $q \in C$ can be decided in 
  $\calO(|\A|^{c\cdot d})$ time. It remains to check whether $p\vec{v}$
  can reach a configuration $q\vec{u}$ above  $|\A|^{c\cdot d}$ via
  a $\Zall$-safe finite path. By Lemma~\ref{lem:cover-short-path} this can be done in time $(|\A|\cdot |\A|^{c\cdot d})^{c'\cdot 2^{d\log(d)}}$ for another constant $c'$. This gives us the desired complexity bound.
\end{proof}

\smallskip

Note that for every fixed dimension $d$, the qualitative 
\mbox{$\Zall$-reachability} problem is solvable in polynomial time.

Now we show that $\calP(\run(p\vec{v},\Zall))$ can be effectively 
approximated up to an arbitrarily small absolute/relative error 
$\varepsilon > 0$. A full proof of Theorem~\ref{thm:approx-general}
can be found in Appendix~\ref{app-approx}.

\begin{theorem}
 \label{thm:approx-general}
 For a given $d$-dimensional pMC $\A$ and its initial configuration
 $p\vec{v}$, the probability $\calP(\run(p\vec{v},\Zall))$ can be
 approximated up to a given absolute error $\eps>0$ in time
 $(\exp(|\A|)\cdot \log(1/\eps))^{\calO(d\cdot d!)}$.
\end{theorem}
\begin{proof}[Proof sketch]
  First we check whether   $\calP(\run(p\vec{v},\Zall))  = 1$ 
  (using the algorithm of Theorem~\ref{thm:qual-all-algorithm}) and 
  return $1$ if it is the case. Otherwise, we first show how
  to approximate $\calP(\run(p\vec{v},\Zall))$ under the 
  assumption that $p$ is in some diverging BSCC of $\F_\A$, and 
  then we show how to drop this assumption.

  So, let $C$ be a diverging BSCC of $\F_\A$ such that 
  $\calP(\run(p\vec{v},C))<1$, and let us assume that $p \in C$. 
  We show how to compute $\nu >0$ such that 
  $|\calP(\run(p\vec{v},\Zall))-\nu|\leq d\cdot\eps$ in time 
  $(\exp(|\A|)\cdot \log(1/\eps))^{\calO(d!)}$. We proceed by induction
  on~$d$. The key idea of the inductive step is to find a sufficiently
  large constant~$K$ such that if some counter reaches~$K$, it can
  be safely ``forgotten'', i.e., replaced by $\infty$, without influencing
  the probability of reaching zero in some counter by more than 
  $\varepsilon$. Hence, whenever we visit a configuration $q\vec{u}$
  where some counter value in $\vec{u}$ reaches $K$, we can
  apply induction hypothesis and approximate the probability or reaching 
  zero in some counter from $q\vec{u}$ by ``forgetting'' the large
  counter a thus reducing the dimension. Obviously, there are only 
  finitely many configurations where all counters are below~$K$, and
  here we employ the standard methods for finite-state Markov chains. 
  The number $K$ is computed by using
  the bounds of Lemma~\ref{lem-divergence}. 

  Let us note that the base (when $d=1$) is handled by relying only
  on Lemma~\ref{lem-divergence}. Alternatively, we could employ
  the results of \cite{ESY:polynomial-time-termination}. This would
  improve the complexity for $d=1$, but not for higher 
  dimensions. 

  Finally, we show how to approximate $\calP(\run(p\vec{v},\Zall))$
  when the control state $p$ does not belong to a BSCC of~$\F_\A$.
  Here we use the bound of Lemma~\ref{lem:F_A-BSCC}.
\end{proof}

Note that if $\calP(\run(p\vec{v},\Zall)) > 0$, then this probability is
at least $\pmin^{m \cdot |Q|}$ where $\pmin$ is the least positive
transition probability in $\M_\A$ and $m$ is the maximal component 
of $\vec{v}$. Hence, Theorem~\ref{thm:approx-general} can also be
used to approximate $\calP(\run(p\vec{v},\Zall))$ up to a given
\emph{relative} error $\varepsilon > 0$.

\subsection{Zero-Reachability, Case~II}
\label{sec-case2}

\newcommand{\calD}{\mathcal{D}}
\newcommand{\totdec}{S^*}

Let us fix a (non-labeled) pMC $\A = (Q,\gamma,W)$ of dimension
$d \in \Nset^+$ and $i\in \{1,\ldots,d\}$.
As in the previous section, our aim is to identify the conditions under which $\run(p\vec{1},\neg\Zminusi{i})>0$. Without restrictions, we assume that $i = d$, i.e., we consider %
$\Zminusi{d} = \{\{1\},\ldots,\{d-1\}\}$.
Also, for technical reasons, we assume that
$\run(p\vec{1},\neg\Zminusi{d})=\run(p\vec{u}^{in},\neg\Zminusi{d})$
where $\uin_i=1$ for all $i\in \{1,\ldots,d-1\}$ but $\vec{u}^{in}_d=0$. (Note that every pMC can be easily modified in polynomial time so that this condition is satisfied.)

To analyze the runs of $\run(p\uin,\neg\Zminusi{d})$, we re-use the finite-state
Markov chain $\F_\A$ introduced in Section~\ref{sec-case1}. Intuitively,
the chain $\F_\A$ is useful for analyzing those runs of
$\run(p\uin,\neg\Zminusi{d})$ where \emph{all} counters stay positive. Since the structure of $\run(p\uin,\neg\Zminusi{d})$
is more complex than in Section~\ref{sec-case1}, we also need some
new analytic tools.

We also re-use the $L$-labeled $1$-dimensional pMC $\B_d$ to deal with
runs that visit zero in counter $d$ infinitely many times. To simplify
notation, we use $\B$ to denote $\B_d$. The behaviour of $\B$ is analyzed
using the finite-state Markov chain $\X$
(see~Definition~\ref{def:X} below) that has been employed already in
\cite{BKK:pOC-time-LTL-martingale} to design a model-checking
algorithm for linear-time properties and one-dimensional pMC.

Let us denote by $[q{\downarrow}r]$ the probability that a run of
$\M_\B$ initiated in $q(0)$ visits the configurations $r(0)$ without
visiting any configuration of the form $r'(0)$ (where $r' \ne r$) in between.
Given $q\in Q$, we denote by $[q{\uparrow}]$ the probability $1-\sum_{r\in Q} [q{\downarrow}r]$ that a run initiated in $q(0)$ never visits a configuration with zero counter value (except for the initial one).
\begin{definition}\label{def:X}
  Let $\X_{\B} = (X,\tran{},\Prob)$ be a non-labelled finite-state Markov
  chain where $X = Q\cup \{q{\uparrow} \mid q \in Q\}$ and
  the~transitions are defined as follows:
  \begin{itemize}\itemsep1ex
  \item $q \tran{x} r$ \ iff \ $0  < x  =  [q{\downarrow}r]$;
  \item $q \tran{x} q{\uparrow}$ \ iff \ $0  < x  = [q{\uparrow}]$;
  \item there are no other transitions.
  \end{itemize}
\end{definition}

\noindent
The correspondence between the runs of $\run_{\M_{\B}}(p(0))$  and
$\run_{\X_{\B}}(p)$ is formally captured by a function
$\Phi : \run_{\M_{\B}}(p(0)) \rightarrow \run_{\X_{\B}}(p) \cup \{\perp\}$, where
$\Phi(w)$ is obtained from a given $w \in \run_{\M_{\B}}(p(0))$
as follows:
\begin{itemize}
\item First, each \emph{maximal} subpath in $w$ of the form
  $q(0),\ldots,r(0)$ such that the counter stays positive in all
  of the intermediate configurations is replaced with a single
  transition $q \tran{} r$.
\item
Note that if $w$ contained infinitely many
  configurations with zero counter, then the resulting sequence
  is a run of $\run_{\X_{\B}}(p)$, and thus we obtain our $\Phi(w)$.
  Otherwise, the resulting sequence takes the form $v\, \hat{w}$, where
  $v \in \fpath_{\X_{\B}}(p)$ and $\hat{w}$ is a suffix
  of~$w$ initiated in a configuration~$r(1)$. Let $q$ be the last state of $v$. Then, $\Phi(w)$ is either $v\,(q{\uparrow})^\omega$ or $\perp$, depending on
  whether $[q{\uparrow}] > 0$ or not, respectively (here,
  $(q{\uparrow})^\omega$ is a infinite sequence of $q{\uparrow}$).
\end{itemize}
\begin{lemma}
\label{lem-measures}
  For every measurable subset $R \subseteq \run_{\X_{\B}}(p)$
  we have that $\Phi^{-1}(R)$ is measurable and
  $\calP(R) = \calP(\Phi^{-1}(R))$.
\end{lemma}
A proof of Lemma~\ref{lem-measures} is straightforward (it suffices to check
that the lemma holds for all basic cylinders $\run_{\X_{\B}}(w)$ where
$w \in \fpath_{\X_{\B}}(p)$). Note that Lemma~\ref{lem-measures} implies
$\calP(\Phi{=}{\perp}) = 0$.

Let $D_1,\ldots,D_k$ be all BSCCs of $\X_{\B}$ reachable from $p$.
Further, for every $D_j$, we use $\run(p\uin,D_j)$ to denote the set
of all $w \in \run_{\M_\A}(p\uin,\neg\Zminusi{d})$ such that
$\Phi(\Upsilon_d(w)) \neq {\perp}$ and $\Phi(\Upsilon_d(w))$ visits $D_j$.
Observe that
\begin{equation}
   \calP(\run_{\M_\A}(p\uin,\neg\Zminusi{d})) \ = \
   \sum_{j=1}^k \calP(\run(p\uin,D_j))
\label{eq-BSCC-D}
\end{equation}
Indeed, note that almost all runs $w$ of $\run_{\X_{\B}}(p)$ visit some $D_j$,
and hence by Lemma~\ref{lem-measures}, we obtain that
$\Phi(w)$ visits some $D_j$ for almost all $w \in \run_{\M_{\B}}(p(1))$.
In particular, for almost all $w$ of $\Upsilon_d(\run_{\M_\A}(p\uin,\neg\Zminusi{d}))$ we have that $\Phi(w)$ visits some $D_j$.
By Lemma~\ref{prop:one-counter-runs}, for almost all $w\in \run_{\M_\A}(p\uin,\neg\Zminusi{d})$, the run $\Phi(\Upsilon_d(w))$ visits some $D_j$, which proves Equation~(\ref{eq-BSCC-D}).

Now we examine the runs of $\run(p\uin,D_j)$ in greater detail and
characterize the conditions under which $\calP(\run(p\uin,D_j)) > 0$.
Note that for every BSCC $D$ in $\X_{\B}$ we have that either $D=\{q{\uparrow}\}$ for some $q\in Q$, or $D\subseteq Q$. We treat these two types of BSCCs separately, starting with the former.

\begin{lemma}
\label{lem-oc-div-prob}
  $\calP(\bigcup_{q\in Q}\run(p\uin,\{q{\uparrow}\})) > 0$ iff there
  exists a~BSCC $C$ of $\F_\A$ with \emph{all} counters diverging and a
  \mbox{$\Zminusi{d}$-safe} finite path of the form
  $p\vec{v} \tran{}^* q\vec{u} \tran{}^* q\vec{z}$ where the subpath
  $q\vec{u} \tran{}^* q\vec{z}$ is $\Zall$-safe,
  $q \in C$, $q\vec{u}$ is above $1$, $\vec{z} - \vec{u} \geq \vec{0}$,
  and $(\vec{z} - \vec{u})[i] > 0$ for every $i$ such that $\vec{t}[i] > 0$.
\end{lemma}
\noindent 
A proof of Lemma~\ref{lem-oc-div-prob} can be found in
Appendix~\ref{app-sec2}. 
Now let $D$ be a BSCC of $\X_{\B}$ reachable from $p$ such that $D\subseteq Q$ (i.e., $D \neq \{q{\uparrow}\}$ for any $q \in Q$).
Let $\vec{e} \in [1,\infty)^D$ where $\vec{e}[q]$ is
the expected number of transitions needed to revisit a configuration with
zero counter from $q(0)$ in $\M_{\B}$.
\begin{proposition}[\cite{BKK:pOC-time-LTL-martingale}, Corollary~6]
The problem whether \mbox{$\vec{e}[q]<\infty$} is decidable in polynomial time.
\end{proposition}
\smallskip

\emph{From now on, we assume that $\vec{e}[q]<\infty$ for all $q\in D$}.
\smallskip

\noindent
In Section~\ref{sec-case1}, we used the trend $\vec{t}\in \Rset^d$ to
determine tendency of counters either to diverge, or to reach zero. As
defined, each $\vec{t}[i]$ corresponds to the long-run average change
per transition of counter~$i$ as long as all counters stay
positive. Allowing zero value in counter $d$, the trend $\vec{t}[i]$
is no longer equal to the long-run average change per transition of
counter~$i$ and hence it does not correctly characterize its
behavior. Therefore, we need to redefine the notion of trend in this case.

Recall that $\B$ is $L=\{-1,0,1\}^{d-1}$-labeled pMC. Given $i\in
\{1,\ldots,d{-}1\}$, we denote by $\vec{\delta}_i\in \Rset^{Q}$ the
vector where $\vec{\delta}_i[q]$ is the $i$-th component of the
expected total reward accumulated along a run from $q(0)$ before
revisiting another configuration with zero counter.
Formally,
$\vec{\delta}_i[q]=\mathbb{E}T_i$ where $T_i$ is a random variable
which to every $w\in \run_{\M_{\B}}(q(0))$ assigns
$\totalrew{i}{w}{0}{\ell}$ such that $\ell>0$ is the least number
satisfying $w(\ell)=r(0)$ for some $r\in D$.

Let $\vec{\mu}_{oc}\in [0,1]^D$ be the invariant distribution of the
BSCC $D$ of $\X_{\B}$, i.e., $\vec{\mu}_{oc}$ is the unique solution
of
\[
\vec{\mu}_{oc}[q] \quad = \quad \sum_{r\in D, r\tran{x} q}  \vec{\mu}_{oc}[r]\cdot x
\]
The \emph{oc-trend} of $D$ is a $(d{-}1)$-dimensional vector $\rt\in
[-1,1]^{d-1}$ defined by
\[
\rt[i]\quad = \quad \left(\vec{\mu}^T_{oc} \cdot \vec{\delta}_i\right)/\left(\vec{\mu}^T_{oc}\cdot \vec{e}\right)
\]
The following lemma follows from the standard results about ergodic
Markov chains (see, e.g., \cite{Norris:book}).

\begin{lemma}\label{lem:long-run-average}
For almost all $w\in \run_{\M_{\B}}(q(0))$ we have that
\[
\rt[i]\quad = \quad \lim_{k\rightarrow \infty} \frac{\totalrew{i}{w}{0}{k}}{k}
\]
\end{lemma}
\noindent
That is, $\rt[i]$ is the $i$-th component of the expected long-run
average reward per transition in a run of $\run_{\M_\B}(q(0))$, and as
such, determines the long-run average change per transition of
counter~$i$ as long as all counters of $\{1,\ldots,d{-}1\}$ remain positive.

Further, for every $i \in \{1,\ldots,d-1\}$ and every $q\in D$, we
denote by $\bottominf_i(q)$ the \emph{least} $j \in \Nset$ such that
every $w \in \fpath_{\M_\B}(q(0))$ ending in $q(0)$ where
$w(n) \neq q(0)$ for all $1 \leq n < \len(w)$ satisfies
$\totalrew{i}{w}{0}{\len(w)} \geq -j$.
If there is no such $j$, we put $\bottominf_i(q) = \infty$.
It is easy to show that if $\bottominf_i(q) = \infty$, then
$\bottominf_i(r) = \infty$ for all $r \in D$.

\begin{lemma}
\label{lem-bottominf}
 If $\bottominf_i(q) < \infty$, then $\bottominf_i(q) \leq 3|Q|^3$ and
 the exact value of $\bottominf_i(q)$ is computable in time polynomial
 in $|\A|$.
\end{lemma}

\noindent
A proof Lemma~\ref{lem-bottominf} can be found in Appendix~\ref{app-sec2}.
We say that counter~$i$ is
{\emph{oc-decreasing}} in $D$ if $\bottominf_i(q) = \infty$ for some (and
hence all) $q \in D$.

\begin{definition}
For a given $i \in \{1,\ldots,d{-}1\}$, we say that the \mbox{$i$-th}
reward is \emph{oc-diverging} in $D$ if either $\rt[i] > 0$, or $\rt[i] = 0$
and counter~$i$ is not oc-decreasing in~$D$.
\end{definition}

\begin{lemma}\label{lem:case2-subcrit-not-diverging}
  If some reward is not oc-diverging in $D$, then $\calP(\run(p\uin,D)) = 0$.
\end{lemma}
A proof of Lemma\ref{lem:case2-subcrit-not-diverging}
can be found in Appendix~\ref{app-sec2}.
It remains to analyze the case when all rewards are \mbox{oc-diverging}
in~$D$. Similarly to Case~I, we need to obtain a bound on probability
of divergence of an arbitrary counter $i \in \{1,\dots,d-1\}$ with
$\rt[i] > 0$.
The following lemma (an analogue of
Lemma~\ref{lem-divergence}) is crucial in the process.

\begin{lemma}
\label{lem:two-counter-divergence}
Let $\calD$ be a $\{-1,0,1\}$-labeled one-dimensional pMC, let $D$ be
a BSCC of $\X_{\calD}$ such that the oc-trend $t_{oc}$ of the only
reward in $D$ is positive. Then for all $q\in D$, there exist
computable constants $h'$ and $A_0$ where $0 < A_0 < 1$, such that for all $h \geq h'$ we have that the probability that
a run $w\in \run_{\M_{\D}}(q(0))$ satisfies
\[
\inf_{k\in \Nset} \totalrew{1}{w}{0}{k}\quad \geq \quad -h
\]
is at least $1- A_0^h$.
\end{lemma}
\noindent

A proof of Lemma~\ref{lem:two-counter-divergence} is the most involved
part of this paper, where we need to construct new analytic tools.
A sketch of the proof is included at the and of this section.

\begin{definition}
  Let $D$ be a BSCC of $\X_{\B}$ where all rewards are oc-diverging,
  and let $q \in D$. We say that a configuration $q\vec{u}$
  is \emph{oc-above} a given $n \in \Nset$ if $\vec{u}[i] \geq n$ for every
  $i\in \{1,\ldots,d-1\}$ %
  such that $\rt[i] > 0$, and
  $\vec{u}[i] \geq \bottominf_i(q)$ for every
  $i\in \{1,\ldots,d-1\}$ %
  such that $\rt[i] = 0$.
\end{definition}

The next lemma is an analogue of Lemma~\ref{lem-diverging} and it is proven
using the same technique, using Lemma~\ref{lem:two-counter-divergence}
instead of Lemma~\ref{lem-divergence}. A full proof can be found
in Appendix~\ref{app-sec2}.

\begin{lemma}
\label{lem-diverging-2}
  Let $D$ be a BSCC of $\X_{\B}$ where all rewards are diverging.
  Then there exists a computable constant~$n \in \Nset$ such that
  $\calP(\run(p\uin,D)) > 0$ iff
  there is a $\Zminusi{d}$-safe finite
  path of the form $p\uin \tran{}^* q\vec{u}$ where $\vec{u}$
  is oc-above $n$ and $\vec{u}[d]=0$.
\end{lemma}

\noindent
A direct consequence of Lemma~\ref{lem-diverging-2} and the results
of \cite{BFLZ:VASSz-model-checking-LMCS} is the following:

\begin{theorem}
  \label{thm:qual-d-algorithm}
  The qualitative $\Zminusi{d}$-reachability problem for \mbox{$d$-dimensional}
  pMC is decidable (assuming $\vec{e}[q]<\infty$ for all $q\in D$ in
  every BSCC of $\X_\B$).
\end{theorem}

\noindent
A proof of Theorem~\ref{thm:qual-d-algorithm} is straightforward, since
we can effectively compute the structure of $\X_\B$ (in time polynomial
in $|\A|$, express its transition probabilities and oc-trends in BSCCs
of $\X_\B$ in the existential fragment of Tarski algebra, an thus
effectively identify all BSCCs of $\X_\B$ where all rewards are
oc-diverging. To check the condition of Lemma~\ref{lem-diverging-2},
we use the algorithm of \cite{BFLZ:VASSz-model-checking-LMCS} for
constructing finite representation of filtered covers in VAS with
one zero test. This is the only part where we miss an upper complexity
bound, and therefore we cannot provide any bound in
Theorem~\ref{thm:qual-d-algorithm}. It is worth noting that
the qualitative $\Zminusi{d}$-reachability problem is
\textsc{Square-Root-Sum}-hard (see below), and hence it  cannot
be solved efficiently without a breakthrough results in the complexity
of exact algorithms. For more comments and a proof of the next
Proposition, see Appendix~\ref{app-sec2}.

\begin{proposition}
\label{sqrt-hard}
  The qualitative $\Zminusi{d}$-reachability problem is
  \textsc{Square-Root-Sum}-hard, even for two-dimensional pMC
  where $\vec{e}[q]<\infty$ for all $q\in D$ in
  every BSCC of $\X_\B$.
\end{proposition}

Using Lemma~\ref{lem-diverging-2}, we can also approximate
$\calP(\run(p\vec{v},\Zminusi{d}))$ up to an arbitrarily small
absolute error $\varepsilon > 0$ (due to the problems mentined
above, we do not provide any complexity bounds). The procedure mimics
the one of Theorem~\ref{thm:approx-general}. The difference is that
now we eventually use methods for one-dimensional pMC instead of the
methods for finite-state Markov chains. The details are given in
Appendix~\ref{sec:case2-approx}.

\begin{theorem}
 \label{thm:approx-general-case2}
 For a given $d$-dimensional pMC $\A$ and its initial configuration
 $p\vec{v}$, the probability $\calP(\run(p\vec{v},\Zminusi{d}))$ can be
 effectively approximated up to a given absolute error $\eps>0$.
\end{theorem}

\smallskip

\noindent
\textbf{A Proof of Lemma~\ref{lem:two-counter-divergence}.}
The  lemma differs from Lemma~\ref{lem-divergence} in that it
effectively bounds the probability of not reaching zero in one of the
counters of a \emph{two-dimensional} pMC (the second counter is
encoded in the labels). Hence, the results on one-dimensional pMCs are
not sufficient here. Below, we sketch a stronger method that allows us
to prove the lemma. The method is again based on analyzing a suitable
martingale; however, the construction and structure of the martingale
is much more complex than in the one-dimensional case.

Before we show how to construct the desired martingale, let us mention
the following useful lemma:

\begin{lemma}
\label{lem:bounded-bumps}
Let $r\in D$. Given a run $w\in \run_{\M_{\B}}(r(0))$, we denote by
$E(w)=\inf\{\ell>0\mid \cval_1(w(\ell))=0\}$, i.e., the time it takes
$w$ to re-visit zero counter value. Then there are constants $c'\in
\Nset$ and $a\in (0,1)$ computable in polynomial space such that for
all $k\geq c'$ we have
\[
\calP(E\geq k)\quad \leq\quad a^k
\]
\end{lemma}
\begin{proof}
This follows immediately from Proposition~6 and Theorem~7 in~\cite{BBEKW:PPDA-time-arXiv}.
\end{proof}

Let us fix an $1$-dimensional pMC $\calD$ with the set of states $Q$ and let us assume, for simplicity, that $\X_\calD$ is strongly connected (assume that the set of states of $\X_\D$ is $D\subseteq Q$). Let us summarize notation used throughout the proof.
\begin{itemize}
\item Let $\vec{e}_\downarrow \in [1,\infty)^Q$ be the vector such that $\vec{e}_\downarrow[q]$
 is the expected total time of a run from $q(1)$ to the first visit of $r(0)$ for some $r \in Q$.
 By our assumptions, $\vec{e}_\downarrow$ is finite.
\item Recall that $\vec{e} \in [1,\infty)^D$ is the vector such that $\vec{e}[q]$
 is the expected total time of a nonempty run from $q(0)$ to the first visit of $r(0)$ for some $r \in Q$.
 Since $\vec{e}_\downarrow$ is finite, also $\vec{e}$ is finite.
\item Let $\vec{\delta}_\downarrow \in \Rset^Q$ be the vector such that $\vec{\delta}_\downarrow[q]$
 is the expected total reward accumulated during a run from $q(1)$ to the first visit of $r(0)$ for some $r \in Q$.
 Since $|\vec{\delta}_\downarrow[q]| \le |\vec{e}_\downarrow[q]|$ holds for all $q \in Q$, the vector $\vec{\delta}_\downarrow$ is finite.
\item Recall that $\vec{\delta}_1 \in \Rset^D$ is the vector such that $\vec{\delta}_1[q]$
 is the expected total reward accumulated during a nonempty run from $q(0)$ to the first visit of $r(0)$ for some $r \in Q$.
 Similarly as before, $\vec{\delta}_1$ is finite.
\item Let $G \in \Rset^{Q \times Q}$ denote the matrix such that $G[q,r]$ is the probability that starting from $q(1)$
 the configuration $r(0)$ is visited before visiting any configuration $r'(0)$ for any $r' \ne r$.
 By our assumptions the matrix~$G$ is stochastic, i.e., $G \vec{1} = \vec{1}$.
\item Let us denote by $A \in \Rset^{D \times D}$ transition matrix of the chain $\X_{\D}$, i.e., $A[q,r]$ is the probability that starting from $q(0)$
 the configuration $r(0)$ is visited before visiting any configuration $r'(0)$ for any $r' \ne r$.
 By our assumptions the matrix~$A$ is stochastic
 and irreducible.
\item Recall that $\vec{\mu}_{oc}^T = \vec{\mu}_{oc}^T A \in [0,1]^D$ denotes the invariant distribution of the finite Markov chain $\X_\D$ induced by~$A$.
\item Recall that $t = (\vec{\mu}_{oc}^T \vec{\delta}_1) / (\vec{\mu}_{oc}^T \vec{e}) \in [-1,+1]$ is the oc-trend of $D$,
 so intuitively $t$ is the expected average reward per step accumulated during a run started from $q(0)$ for some $q \in D$.
\item Let $\vec{r}_\downarrow := \vec{\delta}_\downarrow - t \vec{e}_\downarrow \in \Rset^Q$ and
      let $\vec{r}_0          := \vec{\delta}_1          - t \vec{e} \in \Rset^D$.
\end{itemize}

\begin{lemma} \label{lem-g-zero-exists}
 There exists a vector $\vec{g}(0) \in \Rset^Q$ such that
  \begin{equation}
   \vec{g}(0)[D] = \vec{r}_0 + A \vec{g}(0)[D] \,, \label{eq-mart-fund}
  \end{equation}
   where $\vec{g}(0)[D]$ denotes the vector obtained from~$\vec{g}(0)$ by deleting the non-$D$-components.
\end{lemma}
Extend~$\vec{g}(0)$ to a function $\vec{g} : \Nset \to \Rset^Q$ inductively with
 \begin{equation}
   \vec{g}(n + 1) = \vec{r}_\downarrow + G \vec{g}(n) \qquad \text{for all $n \in \Nset$.} \label{eq-mart-g}
 \end{equation}

\begin{lemma}
\label{lem:weights-bound}
There is $\vec{g}(0)$ satisfying \eqref{eq-mart-fund} for which we have the following:
There exists a constant $c$ effectively computable in polynomial space such that for every $r \in D$ and $n\geq 1$ we have
$|\vec{g}(0)[r]|\leq c $ and $|\vec{g}(n)[r]|\leq c\cdot n$.
\end{lemma}

\newcommand{\hbound}{(t\cdot \sqrt[4]{h})/c}
 Let us fix $q\in D$ and $h\in \Nset$ such that $\hbound\geq c'$, where $c$ is from the previous lemma and $c'$ from Lemma~\ref{lem:bounded-bumps}.
For a run $w \in \run_{\M_\calD}(q(0))$ and all $\ell \in \Nset$ let $\ps\ell \in Q$ and $\xs\ell_1, \xs\ell_2 \in \Nset$
 be such that $\ps\ell = \state(w(\ell))$, $\xs\ell_2 = \cval(w(\ell))$ and $\xs\ell_1 = h+\totalrew{}{w}{0}{\ell}$.

 Now let us define
\begin{equation}
  \ms{\ell} := \xs\ell_1 - t \ell + \vec{g}\big(\xs\ell_2\big)[\ps\ell] \qquad \text{for all $\ell \in \Nset$.}  \label{eq-mart-m}
\end{equation}
Then we have:
\begin{proposition} \label{prop-martingale}
 Write $\calE$ for the expectation with respect to~$\calP$.
 We have for all $\ell \in \Nset$:
 \[
  \calE \left( \ms{\ell+1} \;\middle\vert\; w(\ell) \right) = \ms{\ell}\,.
 \]
\end{proposition}

In other words, the stochastic process $\{\ms{\ell}\}_{\ell = 0}^\infty$ is a martingale. Unfortunately, this martingale may have unbounded differences, i.e. $|\msi{\ell+1}-\msi{\ell}|$ may become arbitrarily large with increasing $\ell$, which prohibits us from applying standard tools of martingale theory (such as Azuma's inequality) directly on  $\{\ms{\ell}\}_{\ell = 0}^\infty$. We now show how to overcome this difficulty.

\newcommand{\ibound}{(t\cdot \sqrt[4]{i})/c}
Let us now fix $i\in \Nset$ such that $i \geq h$ and denote $K=\ibound$. We define a new stochastic process as follows:
\begin{equation}
  \msi{\ell} := \begin{cases}
  	\ms{\ell} & \text{ if } \xs{\ell'}_2\leq K \text{ for all }\ell'\leq \ell\\
  	\msi{\ell-1} & \text{ otherwise. }
  \end{cases}
   \label{eq-mart-mi}
\end{equation}
Observe that $\{\msi{\ell}\}_{\ell = 0}^\infty$ is also a martingale. Moreover, using the bound of Lemma~\ref{lem:weights-bound} we have for every $\ell \in \Nset$ that $|\msi{\ell+1}-\msi{\ell}|\leq 1+t+2cK \leq 4t\sqrt[4]{i}$, i.e., $\{\msi{\ell}\}_{\ell = 0}^\infty$ is a bounded-difference martingale.

Now let $H_i$ be the set of all runs $w$ that satisfy $\xs{i}_1=0$ and $\xs{\ell}_1 > 0$ for all $0 \leq \ell < i$. Moreover, denote by $\Over_i$ the set of all runs $w$ such that $\xs{\ell}_2\geq K$ for some $0\leq \ell \leq i$, and by $\neg\Over_i$ the complement of $\Over_i$.

Note that every run can perform at most $i$-revisits of zero counter value during the first $i$ steps. By Lemma~\ref{lem:bounded-bumps} the probability that counter value at least $ K$ is reached between to visits of zero counter is at most $a^K$. It follows that $\calP(\Over_i)\leq i\cdot a^{\ibound}$.

Next, for every run $w\in \neg\Over_i \cap H_i$ it holds
\begin{align*}
(\msi{i} - \msi{0})(w) &= (\ms{i} - \ms{0})(w) \\
&=- it + \vec{g}(\xs{i}_2)[\ps{i}] - h -\vec{g}(0)[\ps{0}]\\
&\leq -it + 2cK = -it + t\cdot \sqrt[4]{i} \leq -i\frac{t}{2},
\end{align*}

where the first inequality follows from the bound on $\vec{g}(n)$ in Lemma~\ref{lem:weights-bound} and the last inequality holds since $\sqrt[4]{i}\leq i/2$ for all $i\geq 3$.

Using the Azuma's inequality, we get
\begin{align*}
 \calP(\Over_i \cap H_i) &\leq \calP(\msi{i} - \msi{0} \leq -it/2)\\ &\leq \exp\left(-\frac{i^2 \cdot t^2}{8i(4t\sqrt[4]{i})^2} \right)
 =\exp\left(-\frac{\sqrt{i} }{128} \right).
\end{align*}

Altogether, we have
\begin{align*}
 \calP(H_i)&= \calP(H_{i}\cap \Over_i) + \calP(H_i \cap \neg\Over_i) \\
 &\leq i\cdot a^{\ibound} + e^{-\sqrt{i}/128} \leq i\cdot A^{\sqrt[4]{i}},
\end{align*}
where $A=\max \{a^{t/c},2^{-1/128}\}$. Note that $A$ is also computable in polynomial space.

We now have all the tools needed to prove Lemma~\ref{lem:two-counter-divergence}. We have
\begin{align*}
 \calP(\liminf_{k\rightarrow \infty} \totalrew{1}{w}{0}{k}\leq - h)&\leq\calP(\inf_{k\in \Nset} \totalrew{1}{w}{0}{k} \leq -h) \\&= \sum_{i \geq h}\calP(H_i)
 \leq \sum_{i \geq h} i \cdot A^{\sqrt[4]{i}}.
\end{align*}

Note that $\sum_{\ell = h}^{\infty} \ell \cdot A^{\sqrt[4]{\ell}} = \sum_{j =\lfloor\sqrt[4]{h}\rfloor }^{\infty} \sum_{\ell = j^4}^{(j+1)^4 - 1} \ell  \cdot A^{\sqrt[4]{\ell}} \leq \sum_{j =\lfloor\sqrt[4]{h}\rfloor }^{\infty} \sum_{\ell = j^4}^{(j+1)^4 - 1} (j+1)^4 A^{j} \leq \sum_{j =\lfloor\sqrt[4]{h}\rfloor }^{\infty} 8(j+1)^7 A^j$. Using standard methods of calculus we can bound the last sum by  $ (c'' \cdot h^7\cdot A^h)/(1-A)^8$ for some known constant $c''$ independent of $\B$.
Thus, from the knowledge of $A$ and $c''$ we can easily compute, again in polynomial space, numbers $h_0 \in \Nset$, $A_0 \in (0,1)$ such that for all $h \geq h_0$ it holds
\[
 \calP(\liminf_{k\rightarrow \infty} \totalrew{1}{w}{0}{k} \geq h ) \geq 1 - A_0^{h}.
\]

\section{Conclusions}
We have shown that the qualitative zero-reachability problem
is decidable in Case~I and~II, and the probability of all zero-reaching
runs can be effectively approximated. Let us not when the technical condition
adopted in Case~II is not satisfied, than the oc-trends may be undefined
and the problem requires a completely different approach. An important
technical contribution of this paper is the new martingale defined
in Section~\ref{sec-case2}, which provides a versatile tool for attacking
other problems of pMC analysis (model-checking, expected termination time,
constructing (sub)optimal strategies in multi-counter decision 
processes, etc.) similarly as the martingale of 
\cite{BKK:pOC-time-LTL-martingale} for one-dimensional pMC.

\bibliographystyle{abbrv}
\bibliography{str-short,concur,stefan,petri_net}
\appendices
\onecolumn
\section{Proofs of Section~\ref{sec-case1}}
\label{app-sec1}

\begin{reftheorem}{Lemma}{\ref{lem:not-diverging}}
  Let $C$ be a BSCC of $\F_{\A}$.
  If some counter is not diverging in $C$, then $\calP(\run(p\vec{v},C)) = 0$.
\end{reftheorem}
\begin{proof}
  Assume that counter $i$ is not diverging, and consider the
  one-dimensional pMC $\B_i$. Observe that $\F_{\B_i}$ is the same as
  $\F_{\A}$, and hence $\F_{\B_i}$ has the same transition probabilities and
  BSCCs as $\F_{\A}$. In particular, the only counter of $\B_i$ is not 
  diverging in the BSCC $C$ of $\F_{\B_i}$. 
  By the results of~\cite{BKK:pOC-time-LTL-martingale}, almost all runs of
  $\run_{\M_{\B_i}}(p(\vec{v}[i]))$ that stay in $C$ eventually visit
  zero value in the only counter. Since all runs of 
  $\Upsilon_i(\run(p\vec{v},C))$ stay in $C$ but none of them ever
  visits a configuration with zero counter value, we obtain that
  \[
    \calP(\run(p\vec{v},C))=\calP(\Upsilon_i(\run(p\vec{v},C))=0
  \]
\end{proof}

\section{Approximation algorithm for \protect{$\calP(\run(p\vec{v},\Zall))$}}
\label{app-approx}

We show that $\calP(\run(p\vec{v},\Zall))$ can be effectively 
approximated up to an arbitrarily small absolute/relative error 
$\varepsilon > 0$.
First we solve this problem under the assumption that $p$ is in some BSCC of $\F_\A$. Then we show how to drop this assumption.

\begin{proposition}
\label{prop:approx}
There is an algorithm which, for a given \mbox{$d$-dimensional} pMC $\A$, its initial configuration $p\vec{v}$ such that $p$ is in a BSCC of $\F_\A$, and a given $\eps>0$ computes a number $\nu$ such that $|\calP(\run(p\vec{v},\Zall))-\nu|\leq d\cdot\eps$. The algorithm runs in time $(\exp(|\A|)\cdot \log(1/\eps))^{\calO(d!)}$.
\end{proposition}
\begin{proof}
In the following, we denote by $C$ the BSCC of $\A$ containing the initial state $p$.
Note that we may assume that $\calP(\run(p\vec{v},\Zall))<1$. From the proof of Lemma~\ref{lem:cover-short-path} it follows that checking this condition boils down to checking the existence of a certain path of length at most $|\A'|^{\calO(d!)}$ in a suitable VASS $\A'$ of size polynomial in $|\A|$. This can be done it time  $(\exp(|\A|)^{\calO(d!)}$.

We can check this condition using an algorithm of Theorem~\ref{thm:qual-all-algorithm}, and if it does not hold we may output $\nu=1$. In particular, we may assume that the trend of every counter in $C$ is non-negative.

We proceed by induction on $d$. For technical convenience we slightly change the statement about the complexity: we show that the running time of the algorithm is $(\exp(|\A|^c)\cdot\log(\vec{v}_{\max}/\eps))^{d!} $, for some constants $c$, $c'$ independent of $\A$. Clearly, this new statement implies the one in the proposition.

Before we present the algorithm, let us make an important observation. Recall the number $a$ defined in Lemma~\ref{lem-divergence} for an arbitrary one-dimensional pMC $\B$ with a positive trend of the counter. Now suppose that for a given $\B$ and given $\eps>0$ we want to find some $K$ such that $\frac{a^K}{1-a}<\eps$. Note that it suffices to pick any $$K > \frac{\log(1/\eps)}{(1-a)\log(1/a)} .$$ From the definition of $a$ we have $K\in \exp(\B^{O(1)})\cdot\log(1/\eps)$ and that $K$ can be computed in time polynomial in $|\B|$. In particular there is a constant $c$ independent of $\B$ such that $K\leq \exp(|\B|^{c})\cdot\log(1/\eps)$ and we choose $c$ as the desired constant.

Now let us prove the proposition.

$d=1:$ First let us assume that the trend of the single counter in $C$ is $0$. Then, by Lemma~\ref{lem-diverging} it must be the case that $\calP(\run(r(\ell),\Zall)) = 0$ for every $r\in C$ and every $\ell\geq |C|$. Thus, if the initial counter value is $\geq |Q|$, we may output $\nu = 0$. Otherwise, we may approximate the probability by constructing a finite-state polynomial-sized Markov chain $\calM_{|C|}$ whose states are those configurations of $\A$ where the counter is bounded by $|C|$ and whose transitions are naturally derived from $\A$. Formally, $\calM_{|C|}$ is obtained from $\calM_{A}$ by removing all configurations $r(\ell)$ with $\ell > |C|$ and replacing all transitions outgoing from configurations of the form $r(|C|)$ with a self loop of probability 1. Clearly, the value $\calP(\run(p\vec{\ell},\Zall))$ is equal to the probability of reaching a configuration with a zero counter from $p(\ell)$ in $\calM_{|C|}$, which can be computed in polynomial time by standard methods.

If the trend of the counter in $C$ is positive, then let us consider the number $a$ from Lemma~\ref{lem-divergence} computed for $\A$ and $C$. As discussed above, we may compute, in time polynomial in $|\A|$, a number $K \leq\exp(|\A|^c)\cdot\log(1/\eps)$ such that $\frac{a^K}{1-a}<\eps$. We can now again construct a finite-state Markov chain $\calM_{K}$ by discarding all configurations in $\calM_\A$ where the counter surpasses $K$ and replacing the transitions outgoing from configurations of the form $r(K)$ with self-loops.

 Now let us consider an initial configuration $q(\ell)$ with $\ell \leq K$ and denote $P(q(\ell))$ the probability of reaching a configuration with zero counter in from $q(\ell)$ in $\calM_K$. We claim that $|\calP(\run(r(\ell),\Zall)) - P(q(\ell))| \leq \eps$. Indeed, from the construction of $\calM_{K}$ we get that $|\calP(\run(r(\ell),\Zall)) - P(q(\ell))|$ is bounded by the probability, that a run initiated in $q(\ell)$ in $\A$ reaches a configuration of the form $r(K)$ via a $\Zall$-safe 
path \emph{and then} visits a configuration with zero counter. This value is in turn bounded by a probability that a run initiated in $r(K)$ decreases the counter to 0, which is at most $\frac{a^K}{1+a}\leq a^K$ by Lemma~\ref{lem-divergence}, and thus at most $\eps$ by the choice of $K$. Thus, it suffices to compute $P(q(\ell))$ via standard algorithms and return it as $\nu$. 

The same argument shows that if the initial counter value $\ell$ is greater than $K$, we can output $\nu = 0$ as a correct $\eps$-approximation.

Note that the construction of $\calM_K$ and computing the reachability probability in it can be performed in time $(|\A|\cdot K)^{c'}$ for a suitable constant $c'$ independent of $\A$. This finishes the proof of a base case of our induction. 

$d>1:$ Here we will use the algorithm for the $(d-1)$-dimensional case as a sub-procedure. For any counter $i$ and any vector $\vec{\beta}\in\{-1,0,1\}^d$ we denote by $\vec{\beta}_{-i}$ the $(d-1)$-dimensional vector obtained from $\vec{\beta}$ by deleting its $i$-component. Moreover, we define a $(d-1)$-dimensional pMC $\A_{-i}$ obtained from $\A$ by ``forgetting'' the $i$-th counter. I.e., $\A=(Q,\gamma_{-i},W_{-i})$, where $(p,\vec{\alpha},c,q)\in \gamma_{-i}$ iff there is $(p,\vec{\beta},c,q)\in \gamma$ such that $\vec{\beta}_{-i} = \vec{\alpha}$; and where $W_{-i}(p,\vec{\alpha},c,q)=\sum W(p,\vec{\beta},c,q)$ with the summation proceeding over all $\vec{\beta}$ such that $\vec{\beta}_{-i}=\vec{\alpha}$.

Now let us prove the proposition. Let $\vec{t}$ be the trend of $C$. For every counter $i$ such that $\vec{t}[i]>0$ we denote by $a_i$ the number $a$ of Lemma~\ref{lem-divergence} computed for $C$ in $\B_i$ (note that $C$ is a BSCC of every $\B_i$). We put $\amax=\max\{a_i\mid \vec{t}[i]>0\}$. We again compute, as discussed above, in time polynomial in $|\A|$ a number $K \leq\exp(|\A|^c)\cdot\log(1/\eps)$ such that $\frac{\amax^K}{1-\amax}<\eps$. (If $\vec{t}=\vec{0}$, we do not need to define $K$ at all, as will be shown below.) For any configuration $q\vec{u}$ we denote by $\mindiv(q\vec{u})$ the smallest $i$ such that either $\vec{t}[i]>0$ and $\vec{u}[i]\geq K$ or $\vec{t}[i]=0$ and $\vec{u}[i]\geq |C|$ (if such $i$ does not exist, we put $\mindiv(q\vec{u})=\bot$).

Consider a finite-state Markov chain $\calM^d_{K}$ which can be obtained from $\calM_\A$ as follows:
\begin{itemize}
 \item We remove all configurations where at least one of the counters with positive trend is greater than $K$, together with adjacent transitions.
 \item We remove all configurations where at least one of the counters with zero trend is greater than $|C|$, together with adjacent transitions.
 \item We add new states $\qdown$ and $\qup$, both of them having a self-loop as the only outgoing transition.
 \item For every $1\leq i \leq d$ and every remaining configuration ${q\vec{u}}$ with $\mindiv(q\vec{u})=i$ we remove all transitions outgoing from $q\vec{u}$ and replace them with the following transitions:
 \begin{itemize}
 \item A transition leading to $\qdown$, whose probability is equal to some $((d-1)\cdot \eps)$-approximation of $\calP_{\A_{-i}}(\run(q\vec{u}_{-i},\Zall))$ (which can be computed using the algorithm for dimension $d-1$).
 \item A transition leading to $\qup$, with probability $1-x$, where $x$ is such that  $q\vec{u}\tran{x}\qdown$.
 \end{itemize}
\end{itemize}
Above, $\calP_{\A_{-i}}(X)$ represents the probability of event $X$ in pMC $\A_{-i}$. 

Now for an initial configuration $p\vec{v}$ belonging to the states of $\calM^d_K$ let $P(p\vec{v})$ be the probability of reaching, when starting in $p\vec{v}$ in $\calM^d_k$, either the state $\qdown$ or a configuration in which at least one of the counters is 0. Note that $P(p\vec{v})$ can be computed in time polynomial in $|\calM^d_K|$. We claim that $|\calP(\run(p\vec{v},\Zall))-P(p\vec{v})|\leq d\cdot \eps$.

Indeed, let us denote $\Div$ the set of all configurations $q\vec{u}$ such that $q\vec{u}$ is a state of $\calM^d_K$ and $\mindiv({q\vec{u}})\neq \bot$. For every $q\vec{u}\in \Div$ we denote by $x_{q\vec{u}}$ the probability of the transition leading from $q\vec{u}$ to $\qdown$ in $\calM^d_K$. Then $|\calP(\run(p\vec{v},\Zall))-P(p\vec{v})|\leq \max_{q\vec{u}\in\Div}| \calP(\run(q\vec{u},\Zall) - x_{q\vec{u}} |$. Now $\calP(\run(q\vec{u},\Zall)\leq P_1(q\vec{u}) + P_2(q\vec{u})$, where $P_1(q\vec{u})$ is the probability that a run initiated in $q\vec{u}$ in $\A$ visits a configuration with $i$-th counter 0 via a $\Z_{-i}$-safe path, and $P_{2}(q\vec{u})$ is the probability that a run initiated in $q\vec{u}$ in $\A$ visits a configuration with some counter equal to 0 via an $\{i\}$-safe path.

So let us fix $q\vec{u}\in \Div$ and denote $i=\mindiv(q\vec{u})$. If $\vec{t}[i]=0$, then 
we have $P_1(q\vec{u})=0$, since this counter is not decreasing in $C$ and thus it cannot decrease by more than $|C|$. Otherwise $P_1(q\vec{u})$ is bounded by the probability that a run initiated in $q(K)$ in $\B_i$ reaches a configuration where the  counter is 0. From  Lemma~\ref{lem-divergence} we get that $\calP_{\B_{i} }(\run(q(K),\Zall))\leq \frac{a_i^K}{1-a_i} \leq \frac{\amax^K}{1-\amax} \leq \eps$, where the last inequality follows from the choice of $K$. 

For $P_2(q\vec{u})$ note that $P_2{(\vec{u})}=\calP_{\A_{-i}}(\run(q\vec{u}_{-i},\Zall))$ and thus by the construction of $\calM^d_{K}$ we have $|P_2(q\vec{u}) - x_{q\vec{u}}|\leq (d-1)\cdot\eps$.

Altogether we have
\begin{align*}
 &|\calP(\run(p\vec{v},\Zall))-P(p\vec{v})|\leq| P_1(q\vec{u}) + P_2(q\vec{u}) - x_{q\vec{u}}| \leq \eps + (d-1)\cdot \eps = d\cdot \eps.
\end{align*}

Therefore it suffices to compute $P(p\vec{v})$ via standard methods and output is as $\nu$. Finally, if the initial configuration $p\vec{v}$ does not belong to the state space of $\calM^d_K$ let us denote $i=\mindiv(p\vec{v})$. Then it suffices to output some $((d-1)\cdot\eps)$-approximation of $\calP_{\A_{-i}}(\run(p\vec{v}_{-i},\Zall))$ as $\nu$. If $\vec{t}[i]=0$, then $\nu$ is also an $((d-1)\cdot\eps)$-approximation of $\calP_(\run(p\vec{v},\Zall))$, otherwise $|\calP(\run(p\vec{v},\Zall))-\nu|\leq(d-1)\cdot \eps + P_{1}(p\vec{v})$ where $P_1$ is defined in the same way as above. Since the probability of reaching zero counter in $\B_i$ with initial counter value $>K$ can be only smaller than the probability for initial value $K$, the bound on $P_1$ above applies and we get $|\calP(\run(p\vec{v},\Zall))-\nu|\leq d\cdot \eps$.

Now let us discuss the complexity of the algorithm. Note that for any $d$ we have $K \leq \exp(|\A|^c)\cdot\log(1/\eps)$, and the construction of $\calM^d_K$ (or $\calM_K$) and the computation of the reachability probabilities can be done in time $(|\A|\cdot K^{d})^{c'} \cdot T(d-1) \leq (\exp{|\A|^{c+1}}\cdot \log(1/\eps))^{dc'}$ for some  constant $c'$ independent of $\A$ and $d$, where $T(d-1)$ is the running time of the algorithm on a $(d-1)$-dimensional pMC of size $\leq |\A|$ (the pMCs $|\A_{-i}|$ examined during the recursive call of the algorithm are of size $\leq |\A|$). Solving this recurrence we get that the running time of the algorithm is $(\exp(|\A|)\cdot\log(1/\eps))^{\calO(d!)}$.

Lemma~\ref{lem-divergence} we get that $\calP_{\B_{i} }(\run(q(K),\Zall))\leq \frac{a_i^K}{1+a_i} \leq \amax^K \leq \eps$, where the last inequality follows from the choice of $K$. 

\end{proof}

With the help of algorithm from Proposition~\ref{prop:approx} we can easily approximate $\calP(\run(p\vec{v},\Zall))$ even if $p$ is not in any BSCC of $\A$.

\begin{reftheorem}{Theorem}{\ref{thm:approx-general}}
 For a given $d$-dimensional pMC $\A$ and its initial configuration
 $p\vec{v}$, the probability $\calP(\run(p\vec{v},\Zall))$ can be
 approximated up to a given absolute error $\eps>0$ in time
 $(\exp(|\A|)\cdot \log(1/\eps))^{\calO(d\cdot d!)}$.
\end{reftheorem}
\begin{proof}
 First we compute an integer $n\in \exp(|\A|^{\calO(1)})\cdot \log(1/\eps)$ such that $(1 - \pmin^{|Q|})^{\lfloor\frac{n}{|Q|}\rfloor}\leq \eps/2$. This can be done in time polynomial in $|\A|$ and $\log(1/\eps)$. By Lemma~\ref{lem:F_A-BSCC} the probability that a run does not visit, in at most $n$ steps, a configuration $q\vec{u}$ with either $Z(q\vec{u})\neq\emptyset$ or $q$ being in some BSCC of $\A$ is at most $\eps/2$. Now we construct an $n$-step unfolding of $\A$ from $p\vec{v}$, i.e. we construct a finite-state Markov chain $\calM$ such that
\begin{itemize}
 \item its states are tuples of the form $(q\vec{u},j)$, where $ 0\leq j \leq n$ and $ q\vec{u}\text{ is reachable from } p\vec{v} \text{ in $\leq n$ steps in $\A$}$,
 \item for every $0\leq j < n$ we have $(q\vec{u},j)\tran{y}(q'\vec{u}',j+1)$ iff $q\vec{u}\tran{y}q'\vec{u}'$ in $\calM_\A$,
 \item there are no other transitions in $\calM$.
\end{itemize}

  We add to this $\calM$ new states $\qup$ and $\qdown$, and for every state $(q\vec{u},j)$ with $q$ in some BSCC of $\A$ we replace the transitions outgoing from this state with two transitions $(q\vec{u},j)\tran{x}\qdown$, $(q\vec{u},j)\tran{1-x}\qup$, where $x$ is some $(\eps/2)$-approximation of $\calP(\run(q\vec{u},\Zall))$, which can be computed using the algorithm from Proposition~\ref{prop:approx}. Moreover, for every state $(q\vec{u},j)$ with $Z(q\vec{u})\neq \emptyset$ we replace all its outgoing transitions with a single transition leading to $\qdown$. It is immediate that the probability of reaching $\qdown$ from $p\vec{v}$ is an $\eps$-approximation of $\calP(\run(p\vec{v},\Zall))$.

The number of states of $\calM$ is at most $m = n\cdot |Q|\cdot (2n)^d$ and the algorithm of Proposition~\ref{prop:approx} is called at most $m$ times, which gives us the required complexity bound.
\end{proof}

\section{Proofs of Section~\ref{sec-case2}}
\label{app-sec2}

\begin{reftheorem}{Lemma}{\ref{lem-oc-div-prob}}
  $\calP(\bigcup_{q\in Q}\run(p\uin,\{q{\uparrow}\})) > 0$ iff there
  exists a~BSCC $C$ of $\F_\A$ with \emph{all} counters diverging and a
  \mbox{$\Zminusi{d}$-safe} finite path of the form 
  $p\vec{v} \tran{}^* q\vec{u} \tran{}^* q\vec{z}$ where the subpath
  $q\vec{u} \tran{}^* q\vec{z}$ is $\Zall$-safe,
  $q \in C$, $q\vec{u}$ is above $1$, $\vec{z} - \vec{u} \geq \vec{0}$, 
  and $(\vec{z} - \vec{u})[i] > 0$ for every $i$ such that $\vec{t}[i] > 0$.
\end{reftheorem}
\begin{proof}
  ``$\Rightarrow$'' Note that $\calP(\run(p\uin,\{q{\uparrow}\}))>0$ for
  some $q\in Q$.  By Lemma~\ref{lem:F_A-BSCC}, almost every run of
  $\run(p\uin,\{q{\uparrow}\})$ stays eventually in some BSCC of
  $\F_\A$. Let $C$ be a BSCC such that the probability of all 
  $w \in \run(p\uin,\{q{\uparrow}\})$ that stay is $C$ is positive, and let
  $\vec{t}$ be the trend of $C$. We use $R$ to denote the set of all
  $w \in \run(p\uin,\{q{\uparrow}\})$ that stay in $C$. 

  We claim that each counter $i$ must be diverging in $C$.  First, let
  us consider $1\leq i\leq d-1$.  Consider the one-counter pMC
  $\B_i$. Note that the trend of $C$ in $\B_i$ is 
  to $\vec{t}[i]$. For the sake of contradiction, assume that counter~$i$
  is not diverging, i.e., we have either $t_i<0$, or $t_i=0$ and 
  counter~$i$ is decreasing in~$C$. Then, 
  by~\cite{BKK:pOC-time-LTL-martingale}, starting in a
  configuration $p(k)$ of $\B_i$ where $p\in C$, a configuration with
  zero counter value is reached from $p(k)$ with probability
  one. However, then, due to Equation~(\ref{eq-project}) and
  Proposition~\ref{prop:one-counter-runs}, almost every run of $R$
  visits a configuration with zero in one of the counters of
  $\Zminusid$ (note that zero may be reached in some counter before
  inevitably reaching zero in counter~$i$). As $R\subseteq
  \run(p\uin,\{q{\uparrow}\})\subseteq
  \run_{\M_\A}(p\uin,\neg\Zminusi{d})$, we obtain that $\calP(R)=0$, which 
  is a contradiction.  Now consider $i=d$. Similarly as above, starting in
  a configuration $p(k)$ of $\B_d$ where $p\in C$, a configuration
  with zero counter value is reached from $p(k)$ with probability
  one. This implies that almost all runs $w$ of $R$ reach
  configurations with zero counter value in the counter $d$ infinitely
  many times, and hence, by Proposition~\ref{prop:one-counter-runs},
  $\Phi(\Upsilon_d(w))$ does not reach $\bigcup_{q\in Q}
  \{q{\uparrow}\}$ at all. It follows that $\calP(R)=0$,
  a~contradiction.

  Now we prove that for almost all runs $w \in R$ and for all counters
  $i$, one of the following holds:
\begin{enumerate}
\item[(A)] $t_i>0$ and $\liminf_{k\rightarrow \infty}
  \cval_i(w(k))=\infty$,
\item[(B)] $t_i=0$ and $\cval_i(w(k))\geq -\bottomfin_i(\state(w(k)))$
  for all $k$'s large enough.
\end{enumerate}
  The argument is the same as in the proof of Lemma~\ref{lem-diverging}.
  From~(A) and~(B), we immediately obtain the existence of 
  a finite path $p\vec{v} \tran{}^* q\vec{u} \tran{}^* q\vec{z}$
  with the required properties.

  ``$\Leftarrow$'' We argue similarly as in Lemma~\ref{lem-diverging}.
\end{proof}

\begin{reftheorem}{Lemma}{\ref{lem-bottominf}}
 If $\bottominf_i(q) < \infty$, then $\bottominf_i(q) \leq 3|Q|^3$ and
 the exact value of $\bottominf_i(q)$ is computable in time polynomial
 in $|\A|$.
\end{reftheorem}
\begin{proof}[Proof sketch]
 We show that if $\bottominf_i(q) < \infty$, then there is
  $w \in \fpath_{\M_\B}(q(0))$ ending in $q(0)$ where
  $w(n) \neq q(0)$ for all $1 \leq n < \len(w)$,
  $\totalrew{i}{w}{0}{\len(w)} = -\bottominf_i(q)$, and the
  counter is bounded by $2|Q|^2$ along~$w$.
  From this we immediately obtain that $w$ visits at most
  $3|Q|^3$ different configurations, and we can safely assume
  that no configuration is visited twice (if the reward accumulated
  between two consecutive visits to the same configuration is
  non-negative, we can remove the cycle and thus produce a path
  whose total accumulated reward can be only smaller; and if the
  the reward accumulated
  between two consecutive visits to the same configuration is
  negative, we have that $\bottominf_i(q) = \infty$, which is
  a contradiction).

  To see that there is such a path $w$ where the counter is bounded
  by $2|Q|^2$, it suffices to realize that if it was not the case,
  we could always decrease the number of configurations visited by $w$
  where the counter value is above $2|Q|^2$ by removing some
  subpaths of $w$ such that the total reward accumulated in
  these subpaths in non-negative. More precisely, we show
  that there exist configurations $r(i_1)$, $r(i_2)$, $s(i_2)$
  and $s(i_1)$ consecutively visited by $w$ where $0 < i_1 < i_2 \leq 2|Q|^2$,
  the counter stays positive in all configurations between $r(i_1)$
  and $s(i_1)$, the finite path from $r(i_2)$ to $s(i_2)$ visits
  at least one configuration with counter value above $2|Q|^2$,
  and the finite path from $r(i_2)$ to $s(i_2)$ can be ``performed''
  also from $r(i_1)$ without visiting a configuration with zero counter.
  If the total reward accumulated in the paths
  from  $r(i_1)$, $r(i_2)$ and from $s(i_2)$ to $s(i_1)$ is
  negative, we obtain that $\bottominf_i(q) = \infty$ because we can
  ``iterate'' the two subpaths. If it
  is non-negative, we can remove the subpaths from $r(i_1)$ to $r(i_2)$
  and from $s(i_2)$ to $s(i_1)$ from $w$, and thus decrease the
  number of configuration with counter value above $2|Q|^2$, making
  the total accumulated reward only smaller.

  Using the above observations, one can easily compute
  $\bottominf_i(q)$ in polynomial time.
\end{proof}

\begin{reftheorem}{Lemma}{\ref{lem:case2-subcrit-not-diverging}}
  If some reward is not oc-diverging in $D$, then $\calP(\run(p\uin,D)) = 0$.
\end{reftheorem}
\begin{proof}
  Assume that counter $i$ is not diverging in $D$. Let us fix some $q\in
  D$.  Let $w$ be a run in $\M_{\B}$ initiated in $q(0)$ and let
  $I_1<I_2<\cdots$ be non-negative integers such that $w_{I_k}$ is the
  $k$-th occurrence of $q(0)$ in $w$. Given $i\in \{1,\ldots,d-1\}$
  and $k\geq 1$, we denote by
  $T^k_i(w)=\totalrew{i}{w}{I_{k}}{I_{k+1}-1}-\totalrew{i}{w}{I_{k}}{I_{k}}$
  the $i$-th component of the total reward accumulated between the
  $k$-th visit (inclusive) and the $k{+}1$-st visit to $q(0)$
  (non-inclusive).
  We denote by $\mathbb{E}T^k_i$ the expected value of $T^k_i$.

Observe that $T^1_i,T^2_i,\ldots$ are mutually independent and
identically distributed. Thus $T^1_i,T^2_i,\ldots$ determines a random
walk $S^1_i,S^2_i,\ldots$, here $S^k_i=\sum_{j=1}^k T^j_i$, on
$\Zset$. Note that $S^k_i=\totalrew{i}{w}{0}{k+1}$.  By the strong
law of large numbers, for
almost all $w\in \run_{\M_{\B}}(q(0))$,
\begin{eqnarray*}
\mathbb{E}T^1_i & = & \lim_{k\rightarrow \infty} \frac{S^k_i(w)}{k}  \\
    & = & \lim_{k\rightarrow \infty} \frac{S^k_i(w)}{E^k(w)}\frac{E^k(w)}{k} \\
    & = & \lim_{k\rightarrow \infty} \frac{S^k_i(w)}{E^k(w)}\lim_{k\rightarrow\infty}\frac{E^k(w)}{k} \\
    & = & \lim_{k\rightarrow \infty} \frac{\totalrew{i}{w}{0}{k}}{k}\lim_{k\rightarrow \infty} \vec{e}[q]\\
    & = & \vec{t}_{oc}[q] \\
    & \leq & 0
\end{eqnarray*}
(Here $E^k(w)$ denotes the number of steps between the $k$-th and $k+1$-st visit to $q(0)$ in $w$.)
Also, $\calP(T^1_i<0)>0$. By~Theorem~8.3.4 \cite{Chung:book}, for almost all
$w\in \run_{\M_{\B}}(q(0))$ we have that $\liminf_{k\rightarrow
\infty} S^k_i(w)=-\infty$.

However, this also means that almost every run $w\in \run_{\M_{\B}}(q(0))$ satisfies that $\lim_{\ell\rightarrow \infty} \totalrew{i}{w}{0}{\ell}=-\infty$. Subsequently, as all runs of $\Upsilon_d(\run(p\uin,D))$ visit $q(0)$, almost all runs $w$ of $\Upsilon_d(\run(p\uin,D))$ satisfy $\lim_{\ell\rightarrow \infty} \totalrew{i}{w}{0}{\ell}=-\infty$. Thus, by Lemma~\ref{prop:one-counter-runs}, almost all runs of
$\run(p\uin,D)$ visit zero in one of the counters in $\Zminusid$. This means, that $\run(p\uin,D)=0$.
\end{proof}

\begin{reftheorem}{Lemma}{\ref{lem-diverging-2}}
  Let $D$ be a BSCC of $\X_{\B}$ where all rewards are diverging.
  Then there exists a computable constant~$n \in \Nset$ such that
  $\calP(\run(p\uin,D)) > 0$ iff
  there is a $\Zminusi{d}$-safe finite
  path of the form $p\uin \tran{}^* q\vec{u}$ where $\vec{u}$
  is oc-above $n$ and $\vec{u}[d]=0$.
\end{reftheorem}
\begin{proof}
 The constant $n$ is computed using
 Lemma~\ref{lem:two-counter-divergence}. We choose a sufficiently large
 $n$ such that the probability of Lemma~\ref{lem:two-counter-divergence}
 is smaller than $1/d$ for every $q \in D$.

 $\Leftarrow$: Assume that counter $i$ satisfies $\rt[i]>0$.
 By Lemma~\ref{lem:long-run-average}, almost every run $w$ of $\M_\B$ initiated in $q(0)$ satisfies
 \[
 \lim_{k\rightarrow \infty} \totalrew{i}{w}{0}{k}\, /\, k=\rt[i]>0
 \]
 It follows that there is $c>0$ such that for a sufficiently large $k\in \Nset$ we have
 $\totalrew{i}{w}{0}{k}\, /\, k\geq c$.
It follows that $\totalrew{i}{w}{0}{k}\geq ck$ for all sufficiently large $k$. Thus for all counters $i$ satisfying $\rt[i]>0$ and for almost all runs $w$ of $\M_{\B}$ initiated in $q(0)$ we have that $\lim_{k\rightarrow \infty} \totalrew{i}{w}{0}{k}=\infty$.

For every $n\in \Nset$ we denote by $R_n$ the set of all runs $w$ initiated in $q(0)$ such that $\totalrew{i}{w}{0}{k}> -n$ for all $k$ and all $i$ satisfying $\rt[i]>0$. By the above argument, $\calP(\bigcup_{n} R_n)=1$.
Hence, there must be $n$ such that $\calP(R_n)>0$.

Let $q\vec{u}$ be any configuration that is above $n$ and satisfies $\vec{u}[d]=0$. Then $\Upsilon_d(\run(q\vec{u},\Zminusi{d}))\supseteq R_n$ and hence $\calP(\run(q\vec{u},\Zminusi{d}))\geq \calP(R_n)>0$. By our assumption, such a configuration $q\vec{u}$ is reachable from $p\uin$ via a $\Zminusi{d}$-safe path, and thus $\calP(\run(p\uin,D))>0$.

 $\Rightarrow$:
  We show that for almost
  all $w \in \run(p\uin,D)$ and all $i \in \{1,\ldots,d-1\}$, one of
  the following conditions holds:
  \begin{enumerate}
  \item[(A)] $\rt[i]>0$ and $\liminf_{k\rightarrow \infty} \cval_i(w(k))=\infty$,
  \item[(B)] $\rt[i]=0$ and $\cval_i(w(k))\geq \bottominf_i(\state(w(k)))$
    for all $k$'s large enough.
  \end{enumerate}

  Concerning~(A), note that for almost all runs $w$ of $\M_\B$ initiated in $q(0)$ where $q\in D$ we have that
  \[
\lim_{k\rightarrow \infty} \totalrew{i}{w}{0}{k}\, /\, k=\rt[i]>0
\]
which implies, as above, that $\lim_{k\rightarrow \infty} \totalrew{i}{w}{0}{k}=\infty$. Let $q\vec{u}$ be a configuration of $\A$ which is oc-above $1$ and satisfies $\vec{u}[d]=0$.
Then almost all runs $w$ of $\Upsilon_d(\run(q\vec{u},\Zminusi{d}))$ satisfy
$\lim_{k\rightarrow \infty} \totalrew{i}{w}{0}{k}=\infty$, and hence also almost all runs $w$ of $\run(q\vec{u},\Zminusi{d})$ satisfy $\liminf_{k\rightarrow \infty} \cval_i(w(k))=\infty$. As almost every run of $\run(p\uin,D)$ visits $q\vec{u}$ for some $\vec{u}$ that is oc-above $1$ nad satisfying $\vec{u}[d]=0$, almost all runs $w$ in $\run(p\uin,D)$ satisfy
$\liminf_{k\rightarrow \infty} \cval_i(w(k))=\infty$.

  Concerning (B), note that almost all runs $w \in \run(p\uin,D)$
  satisfying $\cval_i(w'(k)) < \bottominf_i(\state(w(k)))$
  for infinitely many $k$'s eventually visit zero in some counter
  (there is a path of length at most $3|Q|^3$ from each such
  $w(k)$ to a configuration with zero in counter $i$, or in one of the
  other counters).

  The above claim immediately implies that for every $n \in \Nset$,
  almost every run of $\run(p\uin,D)$ visits a configuration
  $q\vec{u}$ oc-above~$n$.

  The other implication is proven similarly as in Lemma~\ref{lem-diverging}.
 \end{proof}

\smallskip

Following~\cite{ABKM:Num-Analysis-SIAM} the \textsc{Square-Root-Sum}
problem is defined as follows.  Given natural numbers $d_1, \ldots,
d_n \in \Nset$ and $k \in \Nset$, decide whether $\sum_{i=1}^n
\sqrt{d_i} \ge k$.  Membership of square-root-sum in NP has been open
since 1976.
It is known that \textsc{Square-Root-Sum} 
reduces to PosSLP and hence lies in the counting
hierarchy, see~\cite{ABKM:Num-Analysis-SIAM} and the references
therein for more information on square-root-sum, PosSLP, and the
counting hierarchy.

\begin{reftheorem}{Proposition}{\ref{sqrt-hard}}
  The qualitative $\Zminusi{d}$-reachability problem is
  \textsc{Square-Root-Sum}-hard, even for two-dimensional pMC
  where $\vec{e}[q]<\infty$ for all $q\in D$ in
  every BSCC of $\X_\B$.
\end{reftheorem}
\begin{proof}
We adapt a reduction from~\cite{EWY:one-counter-PE}.
Let $d_1, \ldots, d_n, k \in \Nset$ be an instance of the\textsc{Square-Root-Sum}  problem.
Let $m := \max \{d_1, \ldots, d_n, k\}$.
Define $c_i := \frac12 (1 - d_i/m^2)$ for $i \in \{1, \ldots, n\}$.

We construct a pMC $\A = (Q,\gamma,W)$ as follows.
Take $Q := \{q, r_1, \ldots, r_n, s_+, s_-\}$ and set of rules~$\gamma$ as listed below (we omit labels and some irrelevant rules).
The weight assignment~$W$ is, for better readability, specified in terms of probabilities rather than weights, with the obvious intended meaning.
\begin{align*} \textstyle
\textstyle \frac{1}{2 n} & : (q, (0,0), \emptyset, r_i) && \text{ for all $i \in \{1, \ldots, n\}$} \\
\textstyle \frac{1}{2}   & : (q, (0,-1), \emptyset, s_+) \\
\textstyle \frac{1}{2}   & : (r_i, (0,+1), \emptyset, r_i) && \text{ for all $i \in \{1, \ldots, n\}$} \\
\textstyle c_i           & : (r_i, (0,-1), \emptyset, r_i) && \text{ for all $i \in \{1, \ldots, n\}$} \\
\textstyle \frac12 - c_i & : (r_i, (0,0), \emptyset, s_-) && \text{ for all $i \in \{1, \ldots, n\}$} \\
\textstyle 1             & : (r_i, (0,+1), \{2\}, q) && \text{ for all $i \in \{1, \ldots, n\}$} \\
\textstyle 1             & : (s_-, (0,-1), \emptyset, s_-) \\
\textstyle 1             & : (s_-, (-1,+1), \{2\}, q) \\
\textstyle \frac{k}{n m} & : (s_+, (+1,+1), \{2\}, q) \\
\textstyle 1-\frac{k}{n m} & : (s_+, (0,+1), \{2\}, q)
\end{align*}
We claim that $\calP(\run(q\vec{1},\{\{1\}\})) = 1$ holds if and only if $\sum_i \sqrt{d_i} \ge k$ holds.
It is shown in~\cite{EWY:one-counter-PE} that $r_i(1,1)$ reaches, with probability~$1$, the configuration $r_i(1,0)$ or $s_-(1,0)$
 before reaching any other configuration with $0$ in the second counter.
In fact, it is shown there that the probability of reaching~$s_-(1,0)$ is $\sqrt{d_i}/m$, and of reaching~$r_i(1,0)$ is $1 - \sqrt{d_i}/m$.
The only BSCC~$D$ of~$\X_\B$ is $\{r_1, \ldots, r_n, s_+, s_-\}$.
It follows for the invariant distribution $\vec{\mu}_{oc}$ that
 $\vec{\mu}_{oc}[s_+] = \frac12$ and $\vec{\mu}_{oc}[s_-] = \frac{1}{2 n m} \sum_i \sqrt{d_i}$.
From the construction it is clear that $\vec{\delta}_1[s_+] = +\frac{k}{n m}$ and $\vec{\delta}_1[s_-] = -1$
 and $\vec{\delta}_1[r_i] = 0$ for all $i \in \{1, \ldots, n\}$.
Hence we:
\begin{align*}
\rt[i] & = \left(\vec{\mu}^T_{oc} \cdot \vec{\delta}_i\right)/\left(\vec{\mu}^T_{oc} \cdot \vec{e}\right) \\
 & = \textstyle \left( \frac12 \cdot \frac{k}{n m} - \frac{1}{2 n m} \sum_i \sqrt{d_i} \right) /\left(\vec{\mu}^T_{oc} \cdot \vec{e}\right)
\end{align*}
So we have $\rt[i] \le 0$ if and only if $\sum_i \sqrt{d_i} \ge k$ holds.
The statement then follows from Lemma~\ref{lem-diverging-2}.
\end{proof}

\section{Martingale}

\subsection{Matrix Notation}

In the following, $Q$ will denote a finite set (of control states).
We view the elements of $\Rset^Q$ and $\Rset^{Q \times Q}$ as vectors and matrices, respectively.
The entries of a vector $\vec{v} \in \Rset^Q$ or a matrix $M \in \Rset^{Q \times Q}$ are denoted by $\vec{v}[p]$ and $M[p,q]$ for $p,q \in Q$.
Vectors are column vectors by default; we denote the transpose of a vector~$\vec{v}$ by $\vec{v}^T$, which is a row vector.
For vectors $\vec{u}, \vec{v} \in \Rset^Q$ we write $\vec{u} \le \vec{v}$ (resp.\ $\vec{u} < \vec{v}$)
 if the respective inequality holds in all components.
The vector all whose entries are~$0$ (or~$1$) is denoted by $\vec{0}$ (or~$\vec{1}$, respectively).
We denote the identity matrix by $I \in \{0,1\}^Q$ and the zero matrix by~$0$.
A matrix $M \in [0,1]^{Q \times Q}$ is called \emph{stochastic} (\emph{substochastic}),
 if each row sums up to~$1$ (at most~$1$, respectively).
A nonnegative matrix $M \in [0,\infty)^Q$ is called \emph{irreducible},
 if the directed graph $(Q, \{(p,q) \in Q^2 \mid M[p,q] > 0\})$ is strongly connected.
We denote the spectral radius (i.e., the largest among the absolute values of the eigenvalues) of a matrix~$M$ by~$\rho(M)$.

\subsection{Proof of Lemma~\ref{lem-g-zero-exists}}

The proof of Lemma~\ref{lem-g-zero-exists} is based on the notion of \emph{group inverses} for matrices~\cite{Erdelyi67}.
Close connections of this concept to (finite) Markov chains are discussed in~\cite{Meyer75}.
We have the following lemma:

\begin{lemma} \label{lem-group-inverse}
Let $P$ be a nonnegative irreducible matrix with $\rho(P) = 1$.
Then there is a matrix, denoted by $(I - P)^\#$, such that $(I - P) (I - P)^\# = I - W$,
 where $W$ is a matrix whose rows are scalar multiples of the dominant left eigenvector of~$P$.
\end{lemma}
\begin{proof}
In~\cite{Meyer75} the case of a stochastic matrix~$P$ is considered.
In the following we adapt proofs from~\cite[Theorems 2.1 and~2.3]{Meyer75}.
For a square matrix $M$, a matrix $M^\#$ is called \emph{group inverse} of~$M$, if we have $M M^\# M = M$ and $M^\# M M^\# = M^\#$ and $M M^\# = M^\# M$.
It is shown in~\cite[Lemma~2]{Erdelyi67} that a matrix~$M$ has a group inverse if and only if $M$ and~$M^2$ have the same rank.
As $P$ is irreducible, the Perron-Frobenius theorem implies that the eigenvalue~$1$ has algebraic multiplicity equal to one.
So $0$ is an eigenvalue of~$M := (I-P)$ with algebraic multiplicity~$1$.
This implies that the Jordan form for~$M$ can be written as
\[
 \begin{pmatrix}
  0 & 0 \\
  0 & J'
 \end{pmatrix}
\]
where the square matrix~$J'$ is invertible.
It follows that $M$ and $M^2$ have the same rank, so $M^\#$ exists.
Using the definition of group inverse, we have $(I - M M^\#) P = (I - M M^\#)$.
In other words, the rows of~$I - (I - P) (I - P)^\#$ are left eigenvectors of~$P$ with eigenvalue~$1$.
The statement then follows by the Perron-Frobenius theorem.
\end{proof}

Now we can prove Lemma~\ref{lem-g-zero-exists}.

\begin{proof}
Recall that the matrix~$A$ is stochastic and irreducible.
Also recall from the main body of the paper that $\alpha^T A = \alpha^T$.
It follows from the Perron-Frobenius theorem that $\rho(A) = 1$.
Define $\vec{g}(0)[D] := (I - A)^\# \vec{r}_0$, where $(I - A)^\#$ is the matrix from Lemma~\ref{lem-group-inverse}.
The non-$D$-components can be set arbitrarily, for instance, they can be set to~$0$.
So we have $\vec{g}(0)[D] = \vec{r}_0 + A \vec{g}(0)[D] - W \vec{r}_0$,
 where the rows of~$W$ are multiples of~$\vec{\alpha}^T$.
We have:
\begin{align*}
  \vec{\alpha}^T \vec{r}_0
  & = \vec{\alpha}^T \left( \vec{\delta}_1 - \frac{\vec{\alpha}^T \vec{\delta}_1}{\vec{\alpha}^T \vec{e}} \vec{e} \right)
   && \text{by the definitions of~$\vec{r}_0$ and~$t$} \\
  & = 0\,.
\end{align*}
So \eqref{eq-mart-fund} follows.
\end{proof}

\subsection{Proof of Proposition~\ref{prop-martingale}}

For notational convenience, we assume in the following that $\A$ is a 2-dimensional pMC corresponding to the labelled $1$-dimensional pMC~$\D$
 from the main body; i.e., the first counter of~$\A$ encodes the rewards of~$\D$, the second counter of~$\A$ encodes the unique counter of~$\D$.

Define the substochastic matrices $Q_\rightarrow \in [0,1]^{D \times D}$,  $Q_\uparrow \in [0,1]^{D \times Q}$,
 $P_\downarrow, P_\rightarrow, P_\uparrow \in [0,1]^{Q \times Q}$ as follows:
\begin{align}
 Q_\rightarrow[p,q] & := \sum \{y \mid \exists x_1: p(1,0) \ltran{y} q(x_1,0)\} \label{eq-mart-Q0} \\
 Q_\uparrow[p,q] & := \sum \{y \mid \exists x_1: p(1,0) \ltran{y} q(x_1,1)\} \label{eq-mart-Q+} \\
 P_\downarrow[p,q] & := \sum \{y \mid \exists x_1: p(1,1) \ltran{y} q(x_1,0)\} \label{eq-mart-P-} \\
 P_\rightarrow[p,q] & := \sum \{y \mid \exists x_1: p(1,1) \ltran{y} q(x_1,1)\} \label{eq-mart-P0} \\
 P_\uparrow[p,q] & := \sum \{y \mid \exists x_1: p(1,1) \ltran{y} q(x_1,2)\} \;, \label{eq-mart-P+}
\end{align}
where the transitions $p(1,0) \ltran{y} q(x_1,0)$, etc.\ are in the Markov chain~$\M_\A$.
Note that $Q_\rightarrow + Q_\uparrow$ and $P_\downarrow + P_\rightarrow + P_\uparrow$ are stochastic.
Observe that we have, e.g., that
 $Q_\rightarrow[p,q] = \sum \{y \mid p(0) \ltran{y} q(0)\}$,
 where the transition $p(0) \ltran{y} q(0)$ is in the Markov chain~$\M_\D$.

The matrix~$G$ from the main body of the paper is (see e.g.\ \cite{EWY:one-counter-PE}) the least (i.e., componentwise smallest) matrix with $G \in [0,1]^{Q \times Q}$ and
\begin{equation}
 G = P_\downarrow + P_\rightarrow G + P_\uparrow G G \,. \label{eq-mart-G}
\end{equation}
Recall from the main body that $G$ is stochastic.

For the matrix~$A$ defined in the main body we have
\begin{align}
 A = Q_\rightarrow + Q_\uparrow G[D] \,, \label{eq-mart-A}
\end{align}
where $G[D] \in [0,1]^{Q \times D}$ denotes the matrix obtained from~$G$ by deleting the columns with indices in $Q \setminus D$.
Recall from the main body that $A$ is stochastic and irreducible.

Define
\begin{align}
 B & := P_\rightarrow + P_\uparrow G + P_\uparrow \in [0,1]^{Q \times Q} \,. \label{eq-mart-B}
\end{align}
Define the vectors $\vec{\delta}_{=0!} \in [-1,1]^{D}$, $\vec{\delta}_{>0!} \in [-1,1]^Q$ with
\begin{align}
 \vec{\delta}_{=0!}[p]  & := \sum \{ y x_1 \mid \exists q \in Q \ \exists x_2 : p(1,0) \ltran{y} q(1+x_1, x_2) \} \label{eq-mart-delta-0-bang} \\
 \vec{\delta}_{>0!}[p] & := \sum \{ y x_1 \mid \exists q \in Q \ \exists x_2 : p(1,1) \ltran{y} q(1+x_1, x_2) \} \,, \label{eq-mart-delta-pos-bang}
\end{align}
where the transitions $p(1,0) \ltran{y} q(1+x_1, x_2)$ and $p(1,1) \ltran{y} q(1+x_1, x_2)$ are in the Markov chain~$\M_\A$.
We have that $\vec{\delta}_{=0!}[p]$ is the expected reward incurred in the next step when starting in $p(0)$.
Similarly, $\vec{\delta}_{>0!}[p]$ is the expected reward incurred in the next step when starting in $p(x_2)$ for $x_2 \ge 1$.

\begin{lemma}
The following equalities hold:
\begin{align}
 \vec{e}_\downarrow      & = \vec{1} + B \vec{e}_\downarrow                  \label{eq-mart-t-minus} \\
 \vec{\delta}_\downarrow & = \vec{\delta}_{>0!} + B \vec{\delta}_\downarrow  \label{eq-mart-delta-minus}
\end{align}
\end{lemma}
\begin{proof}
Define the following vectors:
\begin{align*}
 \vec{e}_1 & := P_\downarrow \vec{1} \\
 \vec{e}_2 & := P_\rightarrow (\vec{1} + \vec{e}_\downarrow) \\
 \vec{e}_3 & := P_\uparrow (\vec{1} + \vec{e}_\downarrow) \\
 \vec{e}_4 & := P_\uparrow G \vec{e}_\downarrow
\end{align*}
Observe that $\vec{e}_1 + \vec{e}_2 + \vec{e}_3 + \vec{e}_4$ is the right-hand side of~\eqref{eq-mart-t-minus},
 so we have to show that $\vec{e}_\downarrow = \vec{e}_1 + \vec{e}_2 + \vec{e}_3 + \vec{e}_4$.
Let $q \in Q$.
For concreteness we consider the configuration $q(1)$.
We have that $\vec{e}_1[q]$ is the probability that the first step decreases the counter by~$1$.
Note that we can view $\vec{e}_1[q]$ also as the probability that the first step decreases the counter by~$1$ (namely, to~$0$),
 multiplied with the conditional expected time to reach the $0$-level from~$q(1)$,
 conditioned under the event that the first step decreases the counter by~$1$.
We have that $\vec{e}_2[q]$ is the probability that the first step keeps the counter constant (at~$1$),
 multiplied with the conditional expected time to reach the $0$-level from $q(1)$,
 conditioned under the event that the first step keeps the counter constant.
We have that $\vec{e}_3[q]$ is the probability that the first step increases the counter by~$1$ (namely, to~$2$),
 multiplied with the conditional expected time to reach the $1$-level (again) from $q(1)$,
 conditioned under the event that the first step increases the counter by~$1$.
Finally, $\vec{e}_4[q]$ is the probability that the first step increases the counter by~$1$ (namely, to~$2$),
 multiplied with the conditional expected time to reach the $0$-level \emph{after} having returned to the $1$-level,
 conditioned under the event that the first step increases the counter by~$1$.
So, $\left( \vec{e}_1 + \vec{e}_2 + \vec{e}_3 + \vec{e}_4 \right)[q]$ is the expected time to reach the $0$-level.
Hence \eqref{eq-mart-t-minus} is proved.
The proof of~\eqref{eq-mart-delta-minus} is similar, with reward replacing time.
\end{proof}

By combining \eqref{eq-mart-t-minus} and~\eqref{eq-mart-delta-minus} with the definition of~$\vec{r}_\downarrow$ we obtain:
\begin{equation}
 \vec{r}_\downarrow  = \vec{\delta}_{>0!} - t \vec{1} + B \vec{r}_\downarrow          \label{eq-mart-r-down}
\end{equation}

From the definitions we obtain
\begin{align}
 \vec{\delta}_1 & := \vec{\delta}_{=0!} + Q_\uparrow \vec{\delta}_\downarrow \label{eq-mart-delta-0} \\
 \vec{e}      & := \vec{1} + Q_\uparrow \vec{e}_\downarrow \,. \label{eq-mart-t-0}
\end{align}

By combining \eqref{eq-mart-delta-0} and~\eqref{eq-mart-t-0} with the definition of~$\vec{r}_0$ we obtain:
\begin{equation}
   \vec{r}_0  = \vec{\delta}_{=0!} - t \vec{1} + Q_\uparrow \vec{r}_\downarrow \;. \label{eq-mart-r-0}
\end{equation}

Now we can prove Proposition~\ref{prop-martingale}.
\begin{proof}
We have
\begin{align*}
 & \calE \left( \ms{\ell+1} - \ms\ell \;\middle\vert\; w(\ell), \ \xs\ell_2 = 0 \right) \\
 & = \left( \vec{\delta}_{=0!} - t \vec{1} + Q_\rightarrow \vec{g}(0) + Q_\uparrow \vec{g}(1) - \vec{g}(0)[D] \right)[\ps\ell]
       && \text{by \eqref{eq-mart-m}, \eqref{eq-mart-delta-0-bang}, \eqref{eq-mart-Q0}, \eqref{eq-mart-Q+}} \\
 & = \left( \vec{\delta}_{=0!} - t \vec{1} + Q_\rightarrow \vec{g}(0) + Q_\uparrow \left( \vec{r}_\downarrow + G \vec{g}(0) \right) - \vec{g}(0)[D] \right)[\ps\ell]
       && \text{by \eqref{eq-mart-g}} \\
 & = \left( \vec{r}_{0} + (A - I) \vec{g}(0)[D] \right)[\ps\ell]
       && \text{by \eqref{eq-mart-r-0}, \eqref{eq-mart-A}} \\
 & = 0 && \text{by \eqref{eq-mart-fund}}
\intertext{and}
 & \calE \left( \ms{\ell+1} - \ms\ell \;\middle\vert\; w(\ell), \ \xs\ell_2 > 0 \right) \\
 & = \left( \vec{\delta}_{>0!} - t \vec{1}
     + P_\downarrow \vec{g}\big(\xs\ell_2-1\big)
     + P_\rightarrow \underbrace{\vec{g}\big(\xs\ell_2\big)}_{\stackrel{\eqref{eq-mart-g}}{=}\ \vec{r}_\downarrow + G \vec{g}(\xs\ell_2-1)}
     + \ P_\uparrow \underbrace{\vec{g}\big(\xs\ell_2+1\big)}_{\stackrel{\eqref{eq-mart-g}}{=}\ \vec{r}_\downarrow + G \vec{g}(\xs\ell_2)} - \ \vec{g}(\xs\ell_2) \right)[\ps\ell]
       && \text{by \eqref{eq-mart-m}, \eqref{eq-mart-delta-pos-bang}, \eqref{eq-mart-P-}--\eqref{eq-mart-P+}} \\
 & = \left( \vec{\delta}_{>0!} - t \vec{1} + (P_\rightarrow + P_\uparrow G + P_\uparrow) \vec{r}_\downarrow + \left( P_\downarrow + P_\rightarrow G + P_\uparrow G G \right) \vec{g}\big(\xs\ell_2-1\big) - \vec{g}(\xs\ell_2) \right)[\ps\ell]
       && \text{by \eqref{eq-mart-g}} \\
 & = \left( \vec{r}_\downarrow + G \vec{g}\big(\xs\ell_2-1\big) - \vec{g}(\xs\ell_2) \right)[\ps\ell]
       && \text{by \eqref{eq-mart-B}, \eqref{eq-mart-r-down}, \eqref{eq-mart-G}} \\
 & = 0 && \text{by \eqref{eq-mart-g}\,.}
\end{align*}
\end{proof}

\subsection{Proof of Lemma~\ref{lem:weights-bound}}

Define $\emax := 1 + \max_{q \in Q} \vec{e}_\downarrow[q] \ge 2$.

We first prove the following lemma:
\begin{lemma} \label{lem-mart-potential}
There exists a vector $\vec{g} \in \Rset^D$ with $\vec{g} = \vec{r}_0 + A \vec{g}$
 and
 \[
  0 \le \vec{g}[q] \le \frac{\emax |D|}{\ymin^{|D|}} \qquad \text{for all $q \in D$,}
 \]
 where $\ymin$ denotes the smallest nonzero entry in the matrix~$A$.
\end{lemma}
\begin{proof}
Recall that by Lemma~\ref{lem-g-zero-exists} there is a vector $\vec{g}(0)[D] \in \Rset^D$ with
 \[
  \vec{g}(0)[D] = \vec{r}_0 + A \vec{g}(0)[D] \,.
 \]
Since $A$ is stochastic, we have $A \vec{1} = \vec{1}$.
So there is $\kappa \in \Rset$ such that with $\vec{g} := \vec{g}(0)[D] + \kappa \vec{1}$ we have
\begin{equation}
  \vec{g} = \vec{r}_0 + A \vec{g} \label{eq-mart-potential-proof}
\end{equation}
 and $\gmax = \emax |D| / \ymin^{|D|}$,
 where we denote by $\gmin$ and $\gmax$ the smallest and largest component of~$\vec{g}$, respectively.
We have to show $\gmin \ge 0$.
Let $q \in D$ such that $\vec{g}[q] = \gmax$.
Define the \emph{distance} of a state $p \in D$, denoted by $\eta_p$,
as the distance of~$p$ from~$q$ in the directed graph induced by~$A$.
Note that $\eta_q = 0$
and all $p \in D$ have distance at most $|D|-1$, as $A$ is
irreducible. We prove by induction that a state~$p$ with
distance~$i$ satisfies $\vec{g}[p] \ge \gmax - \emax i/ \ymin^i$.
The claim is obvious for the induction base ($i=0$).  For the induction
step, let $p$ be a state such that~$\eta_p = i+1$. Then there is a
state $r$ such that $A[r,p] > 0$ and $\eta_r = i$.
We have
\begin{align*}
 \vec{g}[r] & =   (A \vec{g})[r] + \vec{r}_0[r]  && \text{by \eqref{eq-mart-potential-proof}} \\
            & \le (A \vec{g})[r] + \emax         && \text{as $\vec{r}_0 \le \emax \vec{1}$} \\
            & = \big(A[r,p] \cdot \vec{g}[p] +
                      \sum_{p' \neq p} A[r,p'] \cdot \vec{g}[p']\big) + \emax \\
            & \le A[r,p] \cdot \vec{g}[p] + (1-A[r,p])\cdot \gmax + \emax && \text{as $A$ is stochastic.}
\end{align*}
By rewriting the last inequality and applying the induction hypothesis
to $\vec{g}[r]$ we obtain
\begin{align*}
  \vec{g}[p] & \geq \gmax - \frac{\gmax - \vec{g}[r] + \emax}{A[r,p]} \geq
  \gmax - \frac{\gmax - (\gmax - \emax i/ \ymin^i) + \emax}{\ymin} \geq
  \gmax - \frac{\emax (i+1)}{\ymin^{i+1}} \,.
\end{align*}
This completes the induction step. Hence we have $\gmin \ge 0$ as desired.
\end{proof}

Now we prove Lemma~\ref{lem:weights-bound}:
\begin{proof}
We need the following explicit expression for~$\vec{g}$:
\begin{equation}
 \vec{g}(n) = G^{n} \vec{g}(0) + \sum_{i=0}^{n - 1} G^i \vec{r}_\downarrow
  \qquad \text{for all $n \ge 0$} \label{eq-mart-g-explicit}
\end{equation}
Let us prove~\eqref{eq-mart-g-explicit} by induction on~$n$.
For the induction base note that the cases $n = 0,1$ follow immediately from the definition~\eqref{eq-mart-g} of~$\vec{g}$.
For the induction step let $n \ge 1$.
We have:
\begin{align*}
 \vec{g}(n + 1)
  & = \vec{r}_\downarrow + G \vec{g}(n) && \text{by~\eqref{eq-mart-g}} \\
  & = G^{n+1} \vec{g}(0) + \sum_{i=0}^{n} G^i \vec{r}_\downarrow  && \text{by the induction hypothesis}
\end{align*}
So \eqref{eq-mart-g-explicit} is proved.
In the following we assume that $\vec{g}(0)$ is chosen as in Lemma~\ref{lem-mart-potential}.
We then have:
\begin{align*}
 |\vec{g}(n)| & \le |\vec{g}(0)| + n |\vec{r}_\downarrow|  && \text{by~\eqref{eq-mart-g-explicit} and as $G$ is stochastic} \\
 & \le \frac{\emax |D|}{\ymin^{|D|}} + n \emax             && \text{by Lemma~\ref{lem-mart-potential} and as $|\vec{r}_\downarrow| \le |\vec{e}_\downarrow| \le \emax$}
\end{align*}
\end{proof}

\section{Proof of Theorem~\ref{thm:approx-general-case2}}
\label{sec:case2-approx}

We show that $\calP(\run(p\vec{v},\Z_{-d}))$ can be effectively 
approximated up to an arbitrarily small absolute error 
$\varepsilon > 0$.

We will use the fact than the probability of reaching a specific set of states in a 1-dimensional pMC can be effectively approximated.

\begin{lemma}
 \label{lem:1pmc-reachability}
Let $\A'$ be any one-dimensional pMC and let $Q$ be its set of states. Given an initial configuration $q(k)$, a set $S\subseteq Q$ and $\eps>0$ we can effectively approximate, up to the absolute error $\eps$, the probability of reaching a configuration $r(j)$ with $r\in S$ from $q(k)$.
\end{lemma}
\begin{proof}
 The crucial observation is that if there is a path from a state $t$ to $S$ in $\F_{\A'}$, then for every $j \geq |Q|$ there is a path of length at most $n$ from $q(j)$ to a configuration with the control state in $S$. If there is no path from $t$ to $S$ in $\F_{\A'}$, then a configuration with the control state in $S$ cannot be reached from $q(j)$ for any $j$. Thus, the probability that a run initiated in $q(k)$ visits a counter value $q(k+i)$ without visiting $S$ \emph{and then} visits $S$ is at most $(1-\pmin^{|Q|})^{\frac{i}{|Q|}}$, where $\pmin$ is the minimal non-zero probability in $\A'$. For a given $\eps$, we can effectively compute $i$ such that $(1-\pmin^{|Q|})^{\frac{i}{|Q|}} \leq \eps$ and effectively construct a finite-state Markov chain $\calM$ in which the configurations of $\A'$ with counter value $\leq i + k$ are encoded in the finite-state control unit (i.e., $\calM$ can be defined as a Markov chain obtained from $\calM_{\A'}$ by removing all configurations with counter height $>i+k$ 
together with their adjacent transitions and replace all transitions outgoing from configurations of the form $r(i+k)$ with self-loops on $r(i+k)$).

Using standard methods for finite-state Markov chains we can compute the probability of reaching the set $S'=\{r(j)\mid r\in S\}$ from $q(k)$ in $\calM$. From the discussion above it follows that this value is an $\eps$-approximation of the probability that $r(j)$ with $r\in S$ is reached in $\A'$.
\end{proof}

The proof closely follows the proof of Theorem~\ref{thm:approx-general}. We first show how to approximate the probability under the assumption that $p$ is in some BSCC $D$ of $\X_\B$. It is then easy to drop this assumption.

\newcommand{\Amax}{A_{\mathit{max}}}
\begin{proposition}
\label{prop:case2-approx}
There is an algorithm which, for a given \mbox{$d$-dimensional} pMC $\A$, its initial configuration $p\vec{v}$ such that $p$ is in a BSCC of $\X_\B$, and a given $\eps>0$ computes a number $\nu$ such that $|\calP(\run(p\vec{v},\Z_{-d}))-\nu|\leq d\cdot\eps$.
\end{proposition}
\begin{proof}
 Clearly we need to consider only $d\geq 2$. We proceed by induction on $d$. The base case and the induction step are solved in almost identical way (which was the case also in the proof of Proposition~\ref{prop:approx}). Therefore, below we present the proof of the induction step and only highlight the difference between the induction step and the base case when needed.

We again assume that $\calP(\run(p\vec{v},\Z_{-d}))<1$. This can be checked effectively due to Theorem~\ref{thm:qual-d-algorithm} and if the condition does not hold, we may output $\nu=1$. In particular we assume that all rewards in $D$ are oc-diverging.

 Recall from the proof of Proposition~\ref{prop:approx} that for any counter $i$ and any vector $\vec{\beta}\in\{-1,0,1\}^d$ we denote by $\vec{\beta}_{-i}$ the $(d-1)$-dimensional vector obtained from $\vec{\beta}$ by deleting its $i$-component; and by $\A_{-i}$ the $(d-1)$-dimensional pMC $\A_{-i}$ obtained from $\A$ by ``forgetting'' the $i$-th counter. (See the proof of Proposition~\ref{prop:approx} for a formal definition). %

Let $\rt$ be the oc-trend of $D$. For every counter $i\in\{1,\dots,d-1\}$ such that $rt[i]>0$ we compute the number $A_0$ of Lemma~\ref{lem-divergence} for $D$ in $\X_\D$, and denote this number by $A_i$. We put $\Amax=\max\{A_i\mid \rt[i]>0\}$. Then we compute a number  such that $\frac{\Amax^K}<\eps/2$. For any $(d-1)$-dimensional vector $\vec{x}$ we denote by $\mindiv(\vec{x})$ the smallest $i\in\{1,\dots,d-1\}$ such that either $\rt[i]>0$ and $\vec{x}[i]\geq K$ or $\vec{t}[i]=0$ and $\vec{x}[i]\geq 3|Q|^3$ (if such $i$ does not exist, we put $\mindiv(\vec{x})=\bot$).

Consider a 1-dimensional pMC $\A_{K} = (Q',\gamma',W')$ which can be obtained from $\A$ as follows:
\begin{itemize}
 \item $Q'$ consists of all tuple $(q,\vec{u})$, where $q\in Q$ and $\vec{u}$ is an arbitrary $(d-1)$-dimensional vector of non-negative integers whose every component is  bounded by $K$; additionally, $Q'$ contains two special states $q{\uparrow}$ and $q{\downarrow}$.
 \item $((q,\vec{u}),j,c,(r,\vec{z}))\in \gamma'$ iff $\mindiv(\vec{u})=\bot$ and $(q,\langle\vec{z}-\vec{u},j\rangle_d,c,r) \in \gamma$. 
 \item For every $1\leq i \leq d-1$ and every $(q,\vec{u})\in |Q|$ such that $\mindiv(\vec{u})\neq \bot$ we have rules $((q,\vec{u}),0,\emptyset,q{\uparrow})$ and  $((q,\vec{u}),0,\emptyset,q{\uparrow})$ in $\gamma'$.
 \item $W'((q,\vec{u}),j,c,(r,\vec{z})) = W(q,\langle\vec{z}-\vec{u},j\rangle_d,c,r)$ for all rules in $\gamma'$ of this shape.
 \item $W'((q,\vec{u}),0,\emptyset,q{\downarrow}) = x$, where $x$ is some $((d-1)\cdot\eps)$-approximation of $\calP_{\A_{-i}}(\run(q\vec{u}_{-i},\Z_{-d}))$ (which can be computed using the algorithm for dimension $d-1$).
 \item $W'((q,\vec{u}),0,\emptyset,q{\uparrow}) = 1-W'((q,\vec{u}),0,\emptyset,q{\downarrow})$.
\end{itemize}

In other words $\A_K$ is obtained from $\A$ by encoding all configurations where all of the first $d-1$ counters are bounded by $K$ explicitly into the state space. If one of these counters surpasses $K$, we ``forget'' about this counter and approximate the 0-reachability in the resulting configuration recursively.

By induction $\A_K$ can be effectively constructed.

Now for an initial configuration $p\vec{v}$ in which the first $d-1$ counters are bounded by $K$ let $P(p\vec{v})$ be the probability of reaching, when starting in $p\vec{v}$ in $\A_K$, either the state $\qdown$ or a state in which at least one of the first $d-1$ counters is 0. Due to Lemma~\ref{lem:1pmc-reachability} we can approximate $P(p\vec{v})$ effectively up to $\eps/2$. We claim that $|\calP(\run(p\vec{v},\Z_{-d}))-P(p\vec{v})|\leq d\cdot \eps$.

Indeed, let us denote $\Div$ the set of all configurations $q\vec{y}$ of $\A$ such that $\vec{y}_{-d}$ is bounded by $K$ and $\mindiv({\vec{y}_{-d}})\neq \bot$. For every $q\vec{u}\in \Div$ we denote by $x_{q\vec{u}}$ the probability of the transition leading from some $(q(k),\vec{u}_{-d})$ to $\qdown$ in $\calM_{\A_K}$ (note that this probability is independent of $k$ and is equal to the weight of the corresponding rule in $\A_K$). Then $|\calP(\run(p\vec{v},\Z_{-d}))-P(p\vec{v})|\leq \max_{q\vec{u}\in\Div}| \calP(\run(q<\vec{u},\Z_{-d}) - x_{q\vec{u}} |$. Now $\calP(\run(q\vec{u},\Z_{-d})\leq P_1(q\vec{u}) + P_2(q\vec{u})$, where $P_1(q\vec{u})$ is the probability that a run initiated in $q\vec{u}$ in $\A$ visits a configuration with $i$-th counter 0 via a $\Z_{-i,d}$-safe path, and $P_{2}(q\vec{u})$ is the probability that a run initiated in $q\vec{u}$ in $\A$ visits a configuration with some counter equal to 0 via an $\{i\}$-safe path.

So let us fix $q\vec{u}\in \Div$ and denote $i=\mindiv(\vec{u}_{-d})$. If $\vec{t}[i]=0$, then 
we have $P_1(q\vec{u})=0$, by Lemma~\ref{lem-bottominf}. Otherwise $P_1(q\vec{u})$ is bounded by the probability that a run $w$ initiated in $q(K)$ in $\B$ satisfies $\inf_{j \geq 0} \totalrew{i}{w}{}{j}\leq -K$ From  Lemma~\ref{lem:two-counter-divergence} we get that this is bounded by $\Amax^K \leq \eps/2$, where the last inequality follows from the choice of $K$. 

For $P_2(q\vec{u})$ note that $P_2{(q\vec{u})}=\calP_{\A_{-i}}(\run(q\vec{u}_{-i},\Z_{-d}))$ and thus by the construction of $\A_{K}$ we have $|P_2(q\vec{u}) - x_{q\vec{u}}|\leq (d-1)\cdot\eps$.

Altogether we have
\begin{align*}
 &|\calP(\run(p\vec{v},\Z_{-d}))-P(p\vec{v})|\leq| P_1(q\vec{u}) + P_2(q\vec{u}) - x_{q\vec{u}}| \leq \eps/2 + (d-1)\cdot \eps.
\end{align*}

Now it is clear that approximating $P(p\vec{v})$ up to $\eps/2$ and returning this value as $\nu$ yields the desired result. As in case 1, if some component of $\vec{v}$ surpasses $K$, we can immediately reduce the problem to the approximation for $(d-1)$-dimensional case.

Note that for the base case $d=2$ the same approach can be used, the only difference that the weight of the rule $((q,\vec{u}),0,\emptyset,q{\uparrow})$ in $\A_K$ is 1 and the weight of $((q,\vec{u}),0,\emptyset,q{\downarrow})$ is 0.

\end{proof}

To prove Theorem~\ref{thm:approx-general-case2} in its full generality it suffices to note, that we can effectively compute a constant $b\in(0,1)$ such that the probability that a run does not visit a configuration $q\vec{u}$ with $q$ in some BSCC of $\X_\B$ or $Z(\vec{u})\neq \emptyset$ in at most $i$ steps is bounded by $b^i$ (see Lemma~\ref{lem:F_A-BSCC} and Lemma~\ref{lem:1pmc-reachability}). Therefore, to approximate the probability for $p\vec{v}$ with $\vec{v}$ not belonging to a BSCC of $\X_\B$ we can use the same approach as in case 1: we unfold $\A$ into a suitable number of steps and approximate the termination value in configurations where the state belongs to some $D$ using the algorithm from the previous proposition. See the proof of Theorem~\ref{thm:approx-general} for further details.

\end{document}